\documentclass[fleqn,10pt]{wlscirep}
\usepackage[utf8]{inputenc}
\usepackage[T1]{fontenc}
\usepackage{graphicx}
\usepackage{url}

\usepackage{dblfloatfix}
\usepackage{amsmath,algorithm}
\usepackage{amsthm}
\usepackage[noend]{algpseudocode}
\usepackage{setspace}
\usepackage{color}
\usepackage{tikz}
\usepackage{pdfpages}
\usetikzlibrary{arrows}
\usepackage{algorithm}
\usepackage{cases}
\usepackage{relsize}
\usepackage{tablefootnote}
\usepackage{threeparttable}
\usepackage{booktabs}
\usepackage{makecell}
\usepackage[markup=hmarkupi]{changes}
\usepackage{amsmath}
\usepackage[noend]{algpseudocode}
\usepackage{cancel}
\usepackage{algpseudocode}% http://ctan.org/pkg/{algorithms,algorithmx}
\algnewcommand{\Inputs}[1]{%
  \State \textbf{Inputs:}
  \Statex \hspace*{\algorithmicindent}\parbox[t]{.8\linewidth}{\raggedright #1}
}
\algnewcommand{\Initialize}[1]{%
  \State \textbf{Initialize:}
  \Statex \hspace*{\algorithmicindent}\parbox[t]{.8\linewidth}{\raggedright #1}
}

\usepackage{algorithm}
\usepackage{algpseudocode}

\usepackage{amssymb}

\makeatletter
\def\BState{\State\hskip-\ALG@thistlm}
\makeatother

\usepackage[T1]{fontenc}
\usepackage[font=small,labelfont=bf,tableposition=top]{caption}

\DeclareCaptionLabelFormat{andtable}{#1~#2  \&  \tablename~\thetable}

\usepackage{graphics}

\usepackage{blindtext}
\usepackage{enumitem}
\usepackage{xcolor}
\usepackage{graphicx}
\usepackage{caption}
\usepackage{ragged2e}
\usepackage[font=small,labelfont=bf,
   justification=justified,
   format=plain]{caption}

\usepackage[T1]{fontenc}
\usepackage[utf8]{inputenc}
\usepackage{babel}
\usepackage[font=small,labelfont=bf]{caption}

\usepackage{tabularx}

\usepackage{lipsum}
\usepackage{mwe}
\usepackage{geometry}
\usepackage{listings}
\usepackage{caption}
\usepackage[export]{adjustbox}
\usepackage{tcolorbox}
\usepackage{multirow}
\usepackage{subcaption}
\usepackage{ragged2e}
\usepackage[T1]{fontenc}
\usepackage[utf8]{inputenc}
\usepackage{babel}
\usepackage[font=small,labelfont=bf]{caption}
\usepackage{booktabs}
\usepackage{amsmath}
\usepackage{multicol}

\usepackage[T1]{fontenc}
\usepackage[english]{babel}
\usepackage{algorithm,algorithmicx,algpseudocode}
\usepackage{amssymb}

\usepackage{calligra}
\DeclareFontFamily{T1}{calligra}{}
\DeclareFontShape{T1}{calligra}{m}{n}{<->s*[1.1] callig15}{}
\DeclareMathAlphabet{\mathcalligra}{T1}{calligra}{m}{n}

\usepackage{mathrsfs}

\algnewcommand{\algorithmicand}{\textbf{ and }}
\algnewcommand{\algorithmicor}{\textbf{ or }}
\algnewcommand{\OR}{\algorithmicor}
\algnewcommand{\AND}{\algorithmicand}
\algnewcommand{\var}{\texttt}

\makeatletter
\newcommand*{\rom}[1]{\expandafter\@slowromancap\romannumeral #1@}
\makeatother
\newcommand{\RNum}[1]{\lowercase\expandafter{\romannumeral #1\relax}}

\DeclareMathOperator{\tr}{tr}

\newcommand{\M}[1]{{\bf #1}}
\newcommand{\V}[1]{{\bf #1}}
\DeclareFontFamily{OT1}{pzc}{}
\DeclareFontShape{OT1}{pzc}{m}{it}{<-> s * [1] pzcmi7t}{}
\DeclareMathAlphabet{\set}{OT1}{pzc}{m}{it}

\newcommand{\T}[1]{\textrm{#1}}

\title{Controllability of complex networks: input node placement restricting the longest control chain}

\author[1]{Samie Alizadeh}
\author[2]{M\'arton P\'osfai}
\author[1,*]{Abdorasoul Ghasemi}
\affil[1]{Dept. of Computer Engineering, K. N. Toosi University of Technology, Tehran, Iran}
\affil[2]{Dept. of Network and Data Science, Central European University, Vienna, Austria}

\affil[*]{corresponding.author arghasemi@kntu.ac.ir}

\begin{abstract}
The minimum number of inputs needed to control a network is frequently used to quantify its controllability.
Control of linear dynamics through a minimum set of inputs, however, often has prohibitively large energy requirements 
and there is an inherent trade-off between minimizing the number of inputs and control energy.
To better understand this trade-off, we study the problem of identifying a minimum set of input nodes such that controllabililty is ensured while restricting the length of the longest control chain.
The longest control chain is the maximum distance from input nodes to any network node, and recent work found that reducing its length significantly reduces control energy.
We map the longest control chain-constraint minimum input problem to finding a joint maximum matching and minimum dominating set.
We show that this graph combinatorial problem is NP-complete, and we introduce and validate a heuristic approximation.
Applying this algorithm to a collection of real and model networks, we investigate how network structure affects the minimum number of inputs, revealing, for example, that for many real networks reducing the longest control chain requires only few or no additional inputs, only the rearrangement of the input nodes.

\end{abstract}
\begin{document}

\flushbottom
\maketitle
% * <samiealizadeh@gmail.com> 2015-02-09T12:07:31.197Z:
%
\thispagestyle{empty}

\section{Introduction}

Network control aims to understand how network structure affects our ability to control the dynamics of complex systems~\cite{liu2016control}.
Development of the field is motivated by a range of applications from biology to epidemiology or technological systems.
Beyond potential applications, designing control strategies that leverage network structure also deepen our understanding of the dynamics of networked systems.
Particularly useful for this are methods that clearly separate the role of network structure from the details of the dynamics on the network, examples of such methods include pinning control, feedback vertex control or structural control of linear systems~\cite{wang2002pinning, fiedler2013dynamics, lin, liu, pang2017universal}.
These methods, however, necessarily make assumptions about the dynamics, limiting the scope of systems they apply to and the range of questions we can ask.
For example, structural controllability of linear systems allows us to determine if a network with a set of input nodes is controllable or not relying on network structure only~\cite{lin}.
However, we cannot use it to design an input signal to actually control the network, nor can we rely on it to calculate the energy needed for control.
To exactly answer these questions, we need a full description of the dynamics, which in the case of linear systems means knowing all link weights exactly. In addition, input node selection that minimizes control energy is an NP-hard problem~\cite{Actuator-Placement}. Most work investigating this problem, therefore, model link weights and largely focus on approximate numerical methods~\cite{li2015minimum,deng2018l0}, for example, using properties of the controllability Gramian matrix that is often ill-conditioned or singular.
To overcome these limitations, several authors developed heuristic methods that rely only on network structure only to investigate control energy related problems, and these heuristics preform reasonably well for realistic networks~\cite{klickstein2022selecting}.
For example, Refs.\cite{LCC,klickstein2018control} found that input node placement that minimizes the longest control chain (LCC), i.e., the largest distance from input nodes to any node, can significantly reduce the energy necessary to control networks.

In this paper, we investigate a graph combinatorial optimization problem motivated by the LCC heuristic.
Specifically, we search for a minimum set of input nodes in a directed network that ensures controllability such that the largest control chain is at most $\ell$.
Relying on the framework of structural controllability, we show that the LCC-constrained minimum input problem is equivalent to solving a combined maximum matching and minimum dominating set problem.
After showing that the problem is NP-complete, we develop a heuristic algorithm to approximate the optimal solution.
We then investigate the effect of network structure on the optimal number of inputs and the cost of reducing the length of the LCC by applying the algorithm to a collection of model and real complex networks.
We demonstrate, for example, that the cost of reducing LCCs diminishes for very heterogeneous networks and very sparse or very dense networks.
We also show that for many real complex networks reducing the LCC compared to the maximum matching based minimum input selection requires only few or no additional inputs, only the rearrangement of the input nodes.

The paper is organized as follows.
Section~\ref{sec: Problem setup and background} summarizes the theoretical background and introduces the problem setup.
Next, in Sec.~\ref{sec:exact_alg}, we formally define the LCC-constrained minimum input problem, we prove that the problem belongs to the NP-complete class and formulate it as an integer linear programming to solve small size instances.
In Sec.~\ref{sec: approx}, we introduce and validate a heuristic algorithm to solve the LCC-constraint input problem for large networks.
We use this algorithm and a collection of model and real networks in Sec.~\ref{sec: result} to study how network structure affects minimum input selection.
Finally, we conclude with a short discussion and outlook in Sec.~\ref{sec:discussion}.

\section{Problem setup and background} \label{sec: Problem setup and background}

Linear time-invariant dynamics with open-loop control became a canonical model to study the control of complex networks, since the choice of linear dynamics often leads to tractable formulation of control problems and it allows focus on the role of network structure\cite{liu}.
To formally define the model, consider a directed network $\set G(\set{V}, \set{E})$ where $\set V=\{v_i | i=1,..., N\}$ is the set of nodes and $\set E \subset \set V \times \set V$ is the set of weighted directed links, i.e., each link $(v_i\rightarrow v_j)$ has an associated weight $a_{ji}\in \mathbb R$ representing the strength of the interaction.
We assign a state $x_i\in \mathbb R$ to each node $v_i$ which evolves following the equation
\begin{equation}
\dot{\V x}(t) = \M A \V x(t) + \M B \V u(t), \label{eq: LTI}
\end{equation}
where the first term on the right hand side represents the internal dynamics of the system and the second term expresses the external control imposed on the network. 
Specifically, $\V x={[x_1, x_2, ..., x_N]}^T$ is the vector of node states, $\M A \in  \mathbb R^{N\times N}$ is the weighted adjacency matrix, $\V u(t)={[u_1, u_2, ..., u_M]}^T\in \mathbb R^M$ is a vector of $M$ time-dependent control signals, and matrix $\M B\in \mathbb R^{N\times M}$ defines how the control signals are coupled to the system.

A system $(\M A,\M B)$ is controllable if, with the proper choice of control signals, we can drive it from any initial state $\V x(0)$ to any final state $\V x(t_\T f)$ in finite time.
Determining controllability, in general, is a numerically unstable problem that requires exact knowledge of the link weights, making it difficult to study for real large-scale networks directly.
However, relying on some mild conditions on the link weights and the structural controllability framework, this inherently numerical problem can be mapped to a combinatorial graph problem.
First, we define an extended network $\set G^\prime(\set V \cup \set U, \set E \cup \set E^\prime)$ where we add additional nodes $\set U=\{u_i | i=1,\ldots,M\}$ representing the control signals and we add a link $(u_i\rightarrow v_j)$ if the control signal $u_i(t)$ is directly connected to network node $v_j$, i.e., $b_{ji}\neq 0$, where $b_{ji}$ is an element of $\M B$.
The network is controllable if in $\set G^\prime$ (i)~all nodes are accessible from the control signals; and (ii)~there exists a matching such that only nodes representing control signals are unmatched~\cite{liu}.
A matching in a directed network is a subset of links $\set E_\T M \subset \set E$ such that no two links in $\set E_\T M$ share start or endpoints, and a node is unmatched if no link in $\set E_\T M$ is pointing at it.

This mapping provides an efficient and robust method to determine the controllability of a network and a controller; however, we are often interested in designing a controller or characterizing control properties of a network without any predefined $\M B$.
For example, the minimum number of input signals $N_\T i$ needed to control a network can serve as a quantity capturing the difficulty of controlling a network.
If a control signal $u_i(t)$ can couple to multiple nodes, finding $N_\T i$ is equivalent to finding a maximum matching in $\set G$: inserting an independent control signal to any unmatched node in $\set G$ ensures that only control signals are unmatched in $\set G^\prime$ and the accessibility condition can be satisfied by adding additional connections between control signals and network nodes as needed.
Finding a maximum matching in a directed network can be done in polynomial time; if, however, each control signal can only be connected to a single network node, i.e., each column of $\M B$ only contains one nonzero element, we cannot satisfy the accessibility condition easily, and the minimum input problem becomes NP-complete~\cite{minimal-controllability-problem}.

Solving the minimum input problem provides a set of nodes to control the system; however, it does not provide any information about the signal $\V u(t)$ or the energy required to control the system.
The control energy of a control signal $\V u(t)$ is defined as
\begin{equation}
\varepsilon = \int_{0} ^ {t_\T f} {\V u_t}^ T \V u_t d_t, \label{eq: required-energy}
\end{equation}
where $t_\T f$ is the time at which a target state is reached~\cite{rughlinear}.
We can express the optimal control signals minimizing the control energy explicitly as
\begin{equation}
\V u (t) = \M B ^ T e ^ {\M A^{T} \left(t_\T f-t\right)} {\M W_\textrm{B}}^ {- 1} \left(t_\T f\right) \left(\V x(t_\T f) -e ^ {\M A t_\T f} \V x(0)\right), \label{eq: optimal-signal}
\end{equation}
where ${\M W_\textrm{B}}^{- 1}(t_\T f)=\int_{0}^{t_\T f} e^ {\M A \tau} \M B \M B^T e^ {\M A^T\tau}d\tau$ is controllability Gramian matrix~ \cite{rughlinear}.
Several factors determine the optimal control energy, including the matrices $\M A$ and $\M B$, the target state $\V x(t_\T f)$, and control time $t_\T f$~\cite{spectrum,controllability-metric,how,kim2018role,lindmark2018minimum}. 
Also, several metrics were proposed to quantify the energy requirements of controlling a system $(\M A,\M B)$, common choices include the average energy or the worst-case energy among all possible target states\cite{controllability-metric}.
Several authors observed that independent of which metric is chosen, the energy required to control networks through minimum input sets is often prohibitively large\cite{LCC,sun2013controllability, spectrum,controllability-metric} and adding additional input nodes can reduce the energy.
This prompts the question: how can we select a small number of input nodes such that both controllability is ensured and control energy is reduced?

Selecting input nodes that minimize the control energy is typically a hard computational problem, and most existing work focuses on numerical approximations that do not explicitly leverage the network structure of the system.
For example, Ref.~\cite{Actuator-Placement} shows that identifying a matrix $\M B$ that minimize the average control energy while keeping $\tr(\M B^T \M B) = M$ fixed is an NP-hard problem.
Reference~\cite{PGM} then propose an approximation algorithm based on the projected gradient method and test it using both synthetic and real networks.
However, almost all elements of the matrix $\M B$ that their algorithm finds are non-zero, meaning that external signals are split up and connected to a significant fraction of network nodes, which is impractical for large networks.
Several extensions were proposed that aim to constrain the number of nodes directly driven by external signals; these methods typically do not guarantee controllability and may suffer performance loss in real networks~\cite{PGM, LPGM, ding2017key}.

The above methods take a generic matrix $\M A$ as input and rely on numerical optimization, not making explicit use of the network structure of the system. 
Acknowledging that $\M A$ represents an adjacency matrix of a complex network allows us to develop methods that use the typical structure of complex networks, potentially leading to useful approximations and providing insights into how these systems function.
For example, for a given network and a set of input nodes, Ref.~\cite{LCC} defined the longest control chain (LCC) as the largest distance between input nodes and any node, i.e.,
\begin{equation}
l_\T{LCC} = \max_{w\in\set V}\min_{v \in \set S} d(v,w),  \label{eq: LCC}
\end{equation}
where $\set S$ is the set of input nodes and $d(v,w)$ is the length of the shortest path connecting from $v$ to $w$.
The authors found that reducing the length of the LCC in a network reduces the energy requirements of control, later theoretical support for this observation was provided~\cite{klickstein2018control,chen2015paradox} and the role of LCC was explored for target control~\cite{klickstein2018energy}.
The relationship between $l_\T{LCC}$ and $\varepsilon$ offers a heuristic strategy to select input nodes that reduce the energy requirements of control:
the algorithms proposed by Refs.~\cite{klickstein2018control,chen2015paradox} take an already controllable system $(\M A, \M B)$ as input and modify $\M B$ by coupling signals to additional nodes or adding further input signals to reduce $l_\T{LCC}$.
Yet, the question of identifying an optimal set of input nodes for a network $\M A$ without first providing a $\M B$ remains unexplored.

In this article, we investigate the LCC-constrained minimum input problem, more specifically, for a given directed network $\M A$ and a integer $\ell>0$, we search for a minimum set of input nodes that simultaneously ensures structural controllability and $\T{LCC}\leq \ell$ such that each input signal is coupled to exactly one node (Fig.~\ref{fig: LCC-constrained controllability}).
Both determining structural controllability of a linear system and calculating the length of the LCC only depend on network structure, hence the LCC-constrained minimum input problem is a graph combinatorial problem.
A central quantity in our study is the minimum number input nodes needed to control a network while ensuring that $l_\T{LCC}\leq \ell$, which we denote by $N_{\textrm{i}}({\ell})$.

In the next section, we formally define the problem and we provide an exact algorithm to solve it using integer linear programming.

\begin{figure}[!h]
\begin{subfigure}[h]{.23\textwidth}
  \centering
  \includegraphics[width=.88\linewidth]{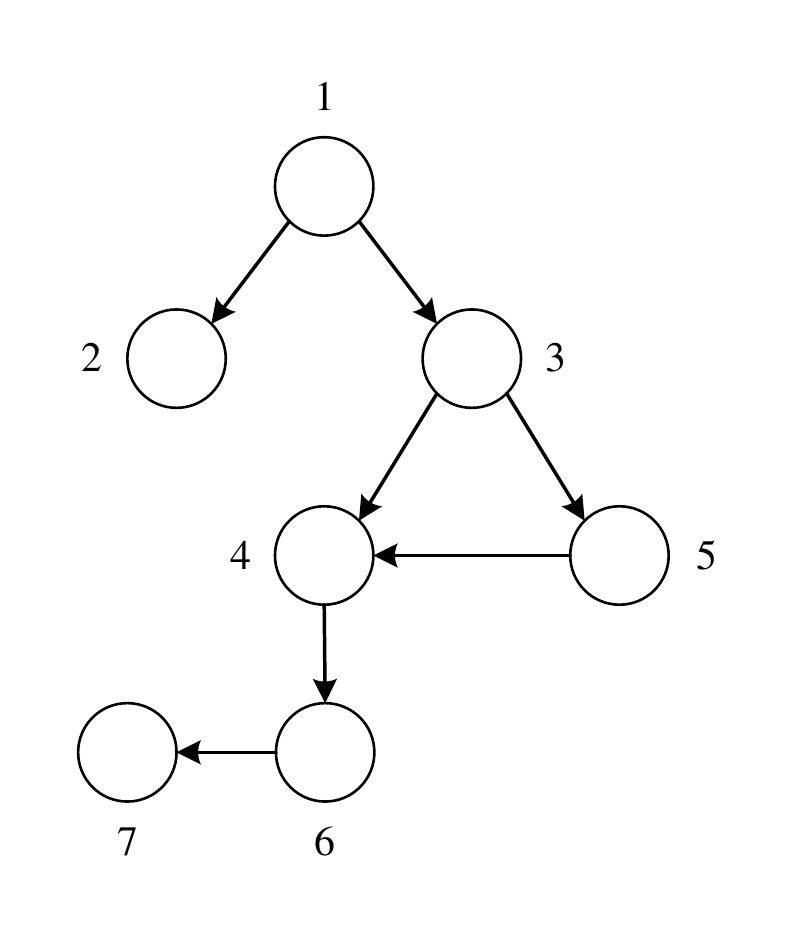}  
\caption{}   
\label{fig: LCC-constrained controllability (a)}
\end{subfigure}
\begin{subfigure}[h]{.23\textwidth}
  \centering
  \includegraphics[width=.85\linewidth]{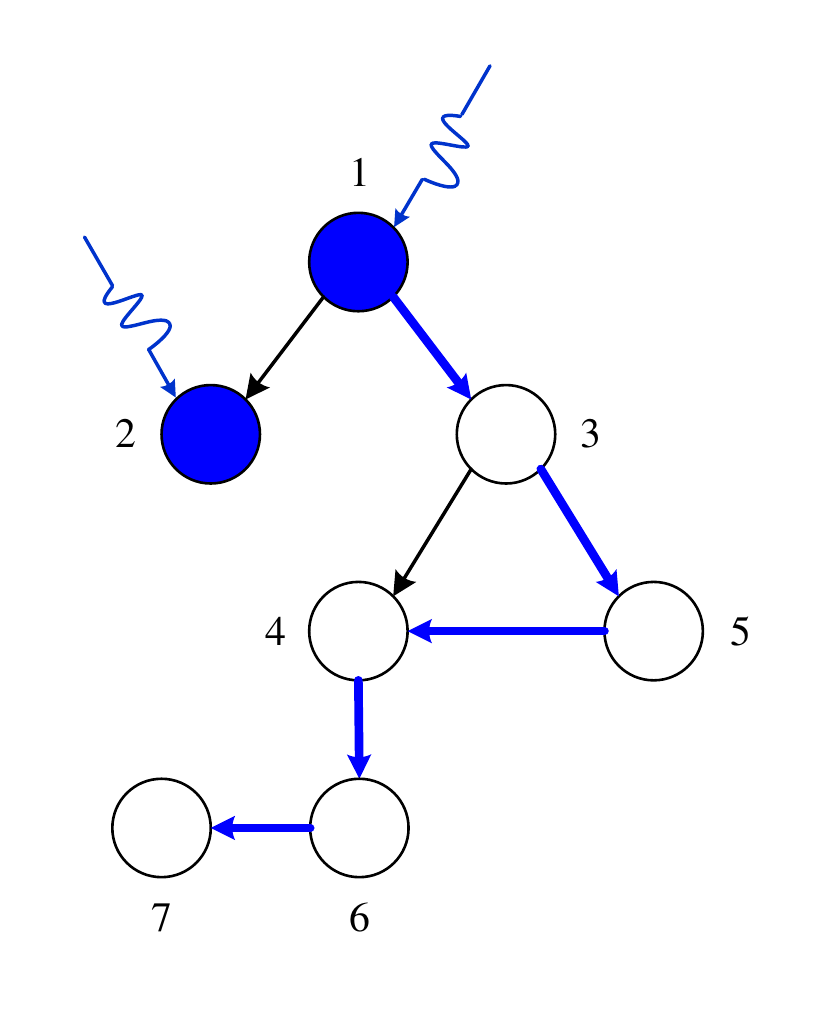}  
\caption{}   
\label{fig: LCC-constrained controllability (b)}
\end{subfigure}
\begin{subfigure}[h]{.18\textwidth}
  \centering
  \includegraphics[width=1.13\linewidth]{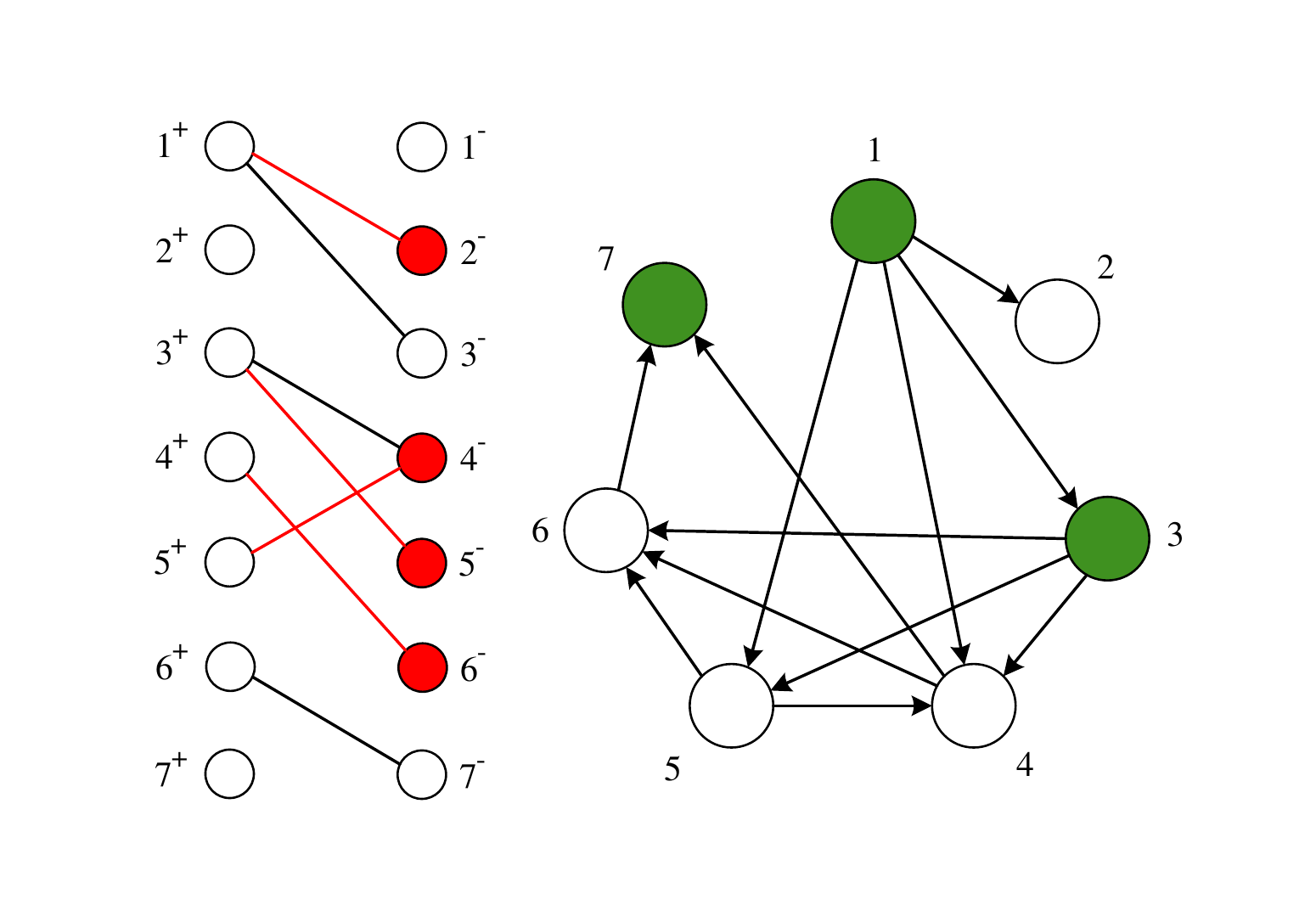}  
\caption{}  
\label{fig: LCC-constrained controllability (d)}
\end{subfigure}
\begin{subfigure}[h]{.30\textwidth}
  \centering
  \includegraphics[width=1.2\linewidth]{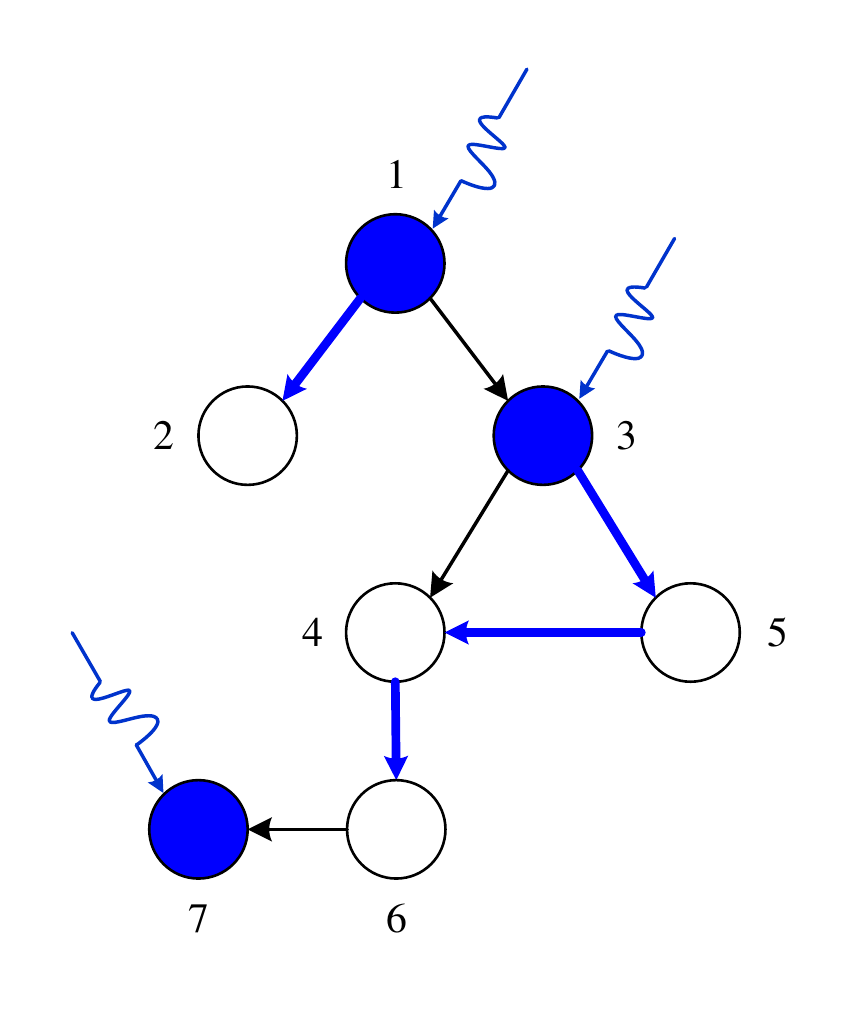}    
\caption{}   
\label{fig: LCC-constrained controllability (c)}
\end{subfigure}
\caption{\textbf{LCC-constrained minimum input problem.} \textbf{(a)}We aim to control a small directed network $\set G$.
\textbf{(b)}~Finding a maximum matching in the network (blue links) and selecting the unmatched nodes 1 and 2 provides a minimum input node set. This input set, however, does not optimize the length of the LCC. In this case the LCC $(1\rightarrow 3\rightarrow 4\rightarrow 6\rightarrow 7)$ has length four. 
\textbf{(c)}~To ensure that $l_\T{LCC} \leq 2$, we need one additional input node (node 7) and we need to control node 3 instead of node 2. 
\textbf{(d)}~The bipartite representation $\set B$ of digraph $\set G$ and the $\ell$-step accessibility graph $\set G_\ell$ for $\ell = 2$ (defined in Sec.~\ref{sec:exact_alg}). Red links in $\set B$ represent a matching. The unmatched nodes in the $\set V^-$ side of $\set B$ form a dominating set in $\set G_\ell$ (green nodes), directly controlling these nodes simultaneously ensures controllability and satisfies the $l_\T{LCC} \leq 2$ condition.}
\label{fig: LCC-constrained controllability}
\end{figure}

\section{LCC-constrained minimum input problem}\label{sec:exact_alg}

We study the problem of identifying a minimum set of input nodes $\set S$ in a directed network $\set G$ such that the system is structurally controllable, the length of the LCC is at most $\ell$, and each input is connected to exactly one node in the network.
To continue, we map the controllability condition to a bipartite matching problem and the LCC constraint to a dominating set problem:  
\begin{itemize}
\item {\bf Controllability constraint:}
The set of input nodes  $\set S$ ensures structural controllability if (i)~there exists a matching such that the set of unmatched nodes is $\set S$ and (ii)~all nodes are accessible from $\set S$.
The condition imposed on the LCC implies accessibility; therefore if $l_\T{LCC} \leq\ell$ then condition (ii) is also satisfied, and we only have to check condition (i).
For this we create a bipartite representation $\set B$ of a directed network $\set G$:  We split each node $v\in \set V$ into two copies $v^+\in \set V^+$ and $v^-\in\set V^-$. If there exists a directed link $(v\rightarrow w)$ in $\set G$, we add an undirected link $(v^+ - w^-)$ to $\set B$.
A directed matching $\set G$ is equivalent to a bipartite matching in $\set B$ (Fig.~\ref{fig: LCC-constrained controllability (c)})~\cite{hopcroft1973n, liu}.

\item {\bf LCC-constraint:} To satisfy the condition $l_\T{LCC}$, we must ensure that there exists a path connecting an input node to each node with length at most $\ell$.
For this we define the $\ell$-step accessibility graph $\set G_\ell$ of a $\set G$, as a directed graph with the same node set as $\set G$, and where we connect node $v$ and $w$ in $\set G_\ell$ with a directed link $(v\rightarrow w)$ if there exists a path from $v$ to $w$ with at most $\ell$ steps in the original network $\set G$.
A node set $\set S$ satisfies the LCC-constraint in $\set G$ if and only if $\set S$ is a dominating set in the accessibility graph $\set G_\ell$, i.e., each node $v$ in $\set G_\ell$ is either a member $\set S$ or there is a link pointing from a node in $\set S$ to $v$ (Fig.~\ref{fig: LCC-constrained controllability (c)}).
\end{itemize}
This mapping allow us to build on algorithms and analytical results originally developed for bipartite matchings and directed dominating sets. 
For example, in the following, we rely on the mapping to calculate a simple upper and lower bound for the minimum number of input nodes $N_\T{i}(\ell)=\lvert \set S \rvert$, and to develop an exact algorithm to solve the minimum input problem.

To find the lower bound, first note that the input nodes $\set S$ are unmatched; therefore $N_{\T i}(\ell)$ is at least equal to the number of unmatched nodes $N_{\T M}$ in a maximum matching of $\set G$.
The input nodes $\set S$ are also a dominating set in $\set G_\ell$; therefore $N_{\T i}(\ell)$ is at least equal to the size of a minimum dominating set $N_\T{DS}$ of $\set G_\ell$. 
These two observations lead to the lower bound
\begin{equation}\label{eq:LB}
N_{\T i}(\ell)\geq \min\{N_\T{M},N_\T{DS}\}.
\end{equation}
To obtain the upper bound, we construct a non-optimal set of input nodes by finding a maximum matching and a minimum dominating set independently, and taking the union of unmatched nodes and the dominating set.
In the worst case there is no overlap between the unmatched and dominating nodes, leading to the upper bound
\begin{equation}\label{eq:UB0}
N_{\T i}(\ell)\leq N_\T{M}+N_\T{DS}.
\end{equation}
We can improve this bound by noting that source nodes, i.e., nodes with no incoming links, are both always unmatched and always dominating.
The source nodes, therefore, are counted twice in Eq.~(\ref{eq:UB0}), correcting for this provides 
\begin{equation}\label{eq:UB}
N_{\T i}(\ell)\leq N_\T{M}+N_\T{DS}-N_\T{s},
\end{equation} 
where $N_\T s$ is the number of source nodes.

To develop an algorithm to find the optimal set of input nodes, we first show that the LCC-constrained minimum input problem is NP-complete by reducing the minimum dominating set problem to it, for details see Sec.~S1 in the Supplementary Information.
This prompts us to reformulate the problem as an integer linear programming (ILP) problem, which allows us to use efficient branch-and-bound ILP solvers to find exact solutions to the minimum input problem.

To define the ILP problem, we assign a binary variable $e_{i \to j} \in \{0, 1\}$ to each link $(i\to j)$, such that $e_{i\to j}=1$ if  link $(i\to j)$ is in the matching, otherwise $e_{i\to j}=0$.
Using this binary vairable we reformulate the LCC-constrained minimum input problem as
\begin{subequations} \label{equ: problem1-ILP}
\begin{align}
& \hspace{3mm} \min_{e_{i \to j} \in \set E} \hspace{5mm} \displaystyle\sum_{v \in \set V} (1- \sum_{i \in {\set{V}_v}^-} e_{i \to v}) \label{equ: problem1-ILP-a} \\
\text{subject to} & \nonumber \\
& \hspace{2mm} \displaystyle\sum_{j \in {\set{V}_v}^+} e_{v \to j} \leq 1 \hspace{25.7mm} \forall \,\, {v \in \set V} \label{equ: problem1-ILP-c}\\
& \hspace{2mm} \displaystyle\sum_{i \in {\set{V}_v}^-} e_{i \to v} \leq 1 \hspace{26.9mm} \forall \,\, {v \in \set V} \label{equ: problem1-ILP-d}\\
& \hspace{2.3mm} \displaystyle\sum_{k \in \set{V}_v^{\ell}} (1- \sum_{i \in {\set{V}_k}^-} e_{i \to k}) \geq 1 \hspace{12.2mm} \forall \,\, {v \in \set V}, \label{equ: problem1-ILP-e}
\end{align}
\end{subequations}
where $\set V^+_i$ and $\set V^-_i$ are the set of neighbors of node $v^+_i$ and $v^-_i$ in $\set B$, and $\set V^{\ell}_i$ is the set of predecessors or in-neighbors of node $v_i$ in the accessibility graph $\set G_\ell$.
The cost function~(\ref{equ: problem1-ILP-a}) counts the number of input nodes, i.e., unmatched nodes in $\set V^-$.
Constraints~(\ref{equ: problem1-ILP-c}) and (\ref{equ: problem1-ILP-d}) enforce the matching criteria, i.e., that each node in $\set B$ is adjacent to at most one link in the matching.
While constraint~(\ref{equ: problem1-ILP-e}) enforces the dominating set criteria, i.e., that each node in $\set G_\ell$ is an input node or has an input node as a predecessor.
Using Eq.~(\ref{equ: problem1-ILP}) we can use off-the-shelf ILP solvers to find exact solutions to the LCC-constraint minimum input problem.
Note, however, that many ILP representations of the same problem may exist and some of the representations might be more suitable for generic ILP solvers.
Indeed, in Sec.~S2 of the Supplementary Information we provide an alternative representation that is more complicated to formulate than Eq.~(\ref{equ: problem1-ILP}), yet leads to superior performance.

Although we rely on efficient ILP solvers, due to hardness of the problem we are able to find minimum input sets for networks with upto a few hundred nodes and links.
We overcome this limitation in the next section by developing a heuristic algorithm and we use the exact algorithm to test the performance of the approximation on small networks.

\section{Approximate algorithm} \label{sec: approx}

Since the LCC-constrained minimum input problem is NP-complete, exact solution is limited to networks with a few hundred links, prompting us to search for approximate solutions.
In this section, we develop a greedy algorithm based on  a set of locally optimal rules.
To find an input set, we iteratively apply these rules to a network: if all nodes are processed we find an  optimal solution, if the process gets stuck, we proceed following a heuristic and the final solution is no longer guaranteed to be optimal. 

The LCC-constrained minimum input problem is a combination of the maximum matching and the minimum dominating set problems.
Next, we review existing greedy methods for these two optimization problems, then we describe how to combine the two algorithms to address the LCC-constrained minimum input problem.

\subsection{Greedy maximum matching}\label{sec:KS_matching}

A greedy construction of a maximal matching is possible using the Karp-Sipser leaf-removal algorithm~\cite{karp1981maximum,zhao2019controllability}, also known as core percolation in the network science literature~\cite{liu2012core,bauer2001core}.
For consistency within the paper, we refer to the algorithm as matching leaf-removal (MLR).
Although polynomial time exact algorithms exist for maximum matching, it is still useful to consider this greedy construction, as it serves as the basis for approximating the NP-complete LCC-constrained minimum input problem.

To apply the MLR to a directed network $\set G$, we first convert $\set G$ to its undirected bipartite network $\set B$.
The MLR algorithm then constructs a maximum matching by iteratively applying the following rule to $\set B$:

\begin{itemize}
\item {\bf Rule-M}: Select a leaf node $v$, i.e., a node with a single neighbor $w$, and place the link $(v-w)$ into the matching.
Remove the link $(v-w)$ from the network together with any other link connected to $w$.
\end{itemize}

If there are no more leaf nodes in the network the MLR halts, and we call the set of remaining links the matching-core (M-core) of the network.
If the M-core of the network is empty, we are done and the obtained matching is optimal.
If the M-core is non-empty, we have to continue relying on a heuristic and the constructed matching is no longer guaranteed to be optimal.
The most simple way to proceed is to add a random link to the matching, remove it together with any adjacent links, and apply the MLR algorithm again to the remaining network. We repeat this until all links are removed.

Note that the order of leaf-removal is not specified by the algorithm and while the size of the M-core, and, if there is no M-core, the cardinality of the obtained matching does not depend on the order of leaf-removal, the specific links that are added to the matching do.
Therefore, different leaf-removal order leads to different matchings even if the M-core is empty and the cardinality of the matching is optimal.
Choosing different leaf-nodes to remove lead to equivalent outcomes when searching for a maximum matching, but this is no longer true when we are simultaneously considering the dominating set problem.

\subsection{Greedy minimum dominating set}\label{sec:GLR_dominating}

Finding the minimum dominating set of a general network is itself NP-complete and therefore a harder problem than maximum matching.
A large number of approximate algorithms exist~\cite{pang2010dominating,molnar2013minimum,zhao2015statistical,simonetti2011minimum,habibulla2017minimal,alipour2020distributed}, here we focus on a greedy generalized leaf-removal algorithm introduced by Habibulla et al., which in spirit is similar to the Karp-Sipser leaf-removal~\cite{habibulla2015directed}.

The generalized leaf-removal scheme, which we refer to as the dominating set leaf-removal (DSLR), iteratively deletes nodes and links from a directed graph $\set G(\set V, \set E)$ to find a minimal set of dominating nodes.
During the DSLR each node $v$ can have three labels: (i)~$v$ is dominating if it has been placed in the dominating set, (ii)~$v$ is observed if $v$ or at least one of its predecessors is a dominating node, or (iii)~otherwise $v$ is unobserved.
Initially all nodes are unobserved and we apply the following rules iteratively:

\begin{itemize}
\item \textbf{Rule-DS1}: If an unobserved node $v$ has no predecessor (i.e. has in-degree zero), add it to the dominating set. All successors of $v$ become observed (Fig. (\ref{fig: DSLR Rules 1})).
\item \textbf{Rule-DS2}: If an unobserved node $v$ has only a single predecessor $w$ and no unobserved successor, add node $w$ to the dominating set. All the successors of $w$ then become observed (Fig. (\ref{fig: DSLR Rules 2})).
\item \textbf{Rule-DS3}: If a node $v$ is observed and has only a single unobserved successor $w$, selecting $v$ as a dominating node is not better than selecting $w$, therefore the link $(v, w)$ is deleted from the network. The node $w$ remains unobserved after the link deletion (Fig. (\ref{fig: DSLR Rules 3})).
\end{itemize}

\begin{figure}[!h]
\begin{subfigure}[h]{.33\textwidth}
  \centering
  \includegraphics[width=.55\linewidth]{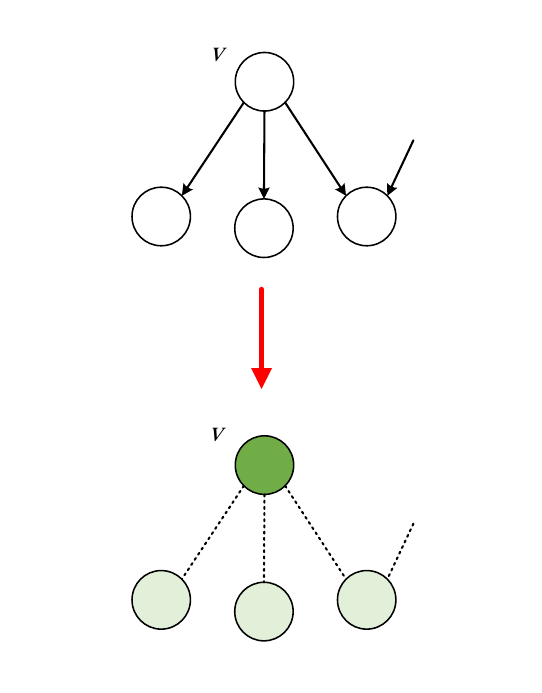}  
\caption{Rule-DS1}   
\label{fig: DSLR Rules 1}
\end{subfigure}
\begin{subfigure}[h]{.33\textwidth}
  \centering
  \includegraphics[width=.55\linewidth]{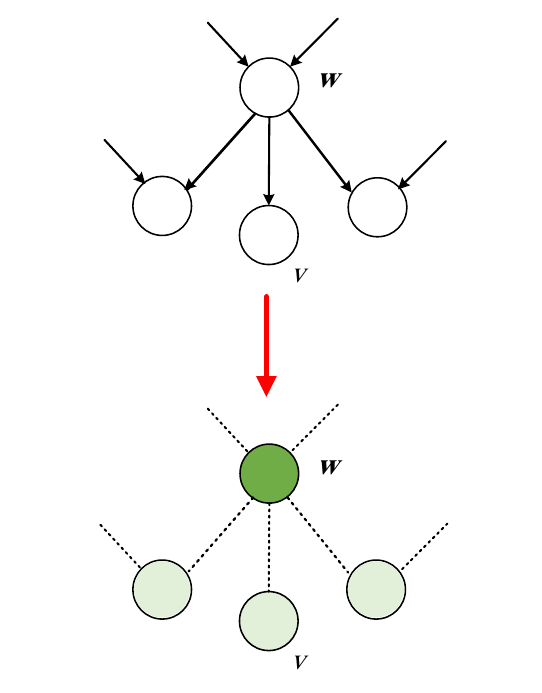}  
\caption{Rule-DS2}   
\label{fig: DSLR Rules 2}
\end{subfigure}
\begin{subfigure}[h]{.33\textwidth}
  \centering
  \includegraphics[width=.55\linewidth]{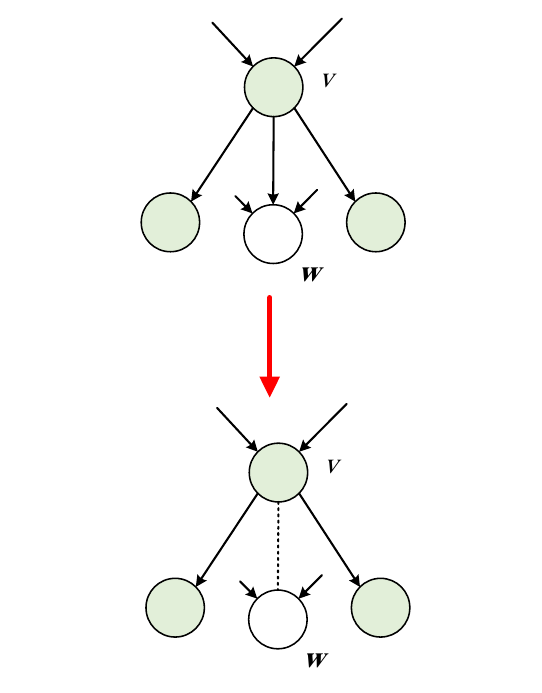}    
\caption{Rule-DS3} 
\label{fig: DSLR Rules 3}
\end{subfigure}
\caption{\textbf{Local rules to construct a minimal dominating set. } White nodes are unobserved, light green nodes are observed, and dark green indicates dominating nodes. Removed links are drawn with a dashed line. }
\label{fig: DSLR Rules}
\end{figure}

When applying the rules, links that point to dominating nodes and their successors are removed, i.e., incoming links to observed nodes.  
If no more links can be removed from the network, we call the remaining set of links the dominating set-core (DS-core) of the network.
Reference~\cite{habibulla2015directed} showed that if the DS-core of the network is empty, the obtained dominating set (DS) is optimal.
If the DS-core is non-empty, we must proceed relying on a heuristic and the solution is no longer guaranteed to be opitmal.
Specifically, we continue  by adding the node with the highest degree to the dominating set, setting all of its successors to be observed, then applying the the DSLR rules to the updated network. We repeat this until all links are removed.

Similarly to maximum matching, the dominating set produced by the DSLR algorithm is not unique, even if the network has no DS-core.
The choice of the order that the rules are applied to the network does not affect the outcome when minimizing the dominating set, but again, this is no longer true if we are simultaneously aiming to satisfy the matching problem.

\subsection{LCC-constrained minimum input problem}

To identify an LCC-constrained minimum input set, we search for a set of nodes $\set S$ such that both (i)~$\set S$ is a dominating set in the accessibility graph $\set G_{\ell}$ and (ii)~there exists a matching in $\set B$ such that $\set S$ is the set of unmatched nodes in $\set V^-$.
Our strategy is to use the MLR to identify unmatched nodes in $\set B$ and update $\set G_{\ell}$ by adding these nodes to the dominating set.
If we get stuck, we proceed by applying the DSLR algorithm to $\set G_{\ell}$ and update $\set B$.
We iterate these steps until both networks are processed, if both MLR and DSLR get stuck, we rely on heuristics to proceed.

The construction of the matching and the dominating set are coupled together, i.e., every time, we modify $\set B$ we need to update $\set G_{\ell}$ and vice versa:
\begin{itemize}
\item \textbf{$\set B \rightarrow \set G_{\ell}$}: If a node $v \in \set {V^-}$ becomes isolated in $\set B$, it becomes unmatched and hence a control node. Therefore the corresponding  $v$ node in $\set G_l$ is added to the dominating set, any successor $w$ of $v$ becomes observed, and all links connected to $v$ and any incoming links connected to $w$ are deleted.
\item \textbf{$\set G_{\ell} \rightarrow \set B$}: If a node $v$ is a dominating node in $\set G_{\ell}$, it has to be a control node.
Therefore the corresponding node $v^-$ in $\set B$ does not have to be matched, hence we remove any links connected to $v^-$.
\end{itemize}

Due to the coupling between the matching and dominating set problems, we cannot simply apply the rules described in Secs.~\ref{sec:KS_matching} and \ref{sec:GLR_dominating} to $\set B$ and $\set G_{\ell}$ independently, since, for example, the order that rule-M is applied to $\set B$ may produce different outcomes in $\set G_{\ell}$.
Instead, when applying the local rules, we must take into account their effect on both networks, and we only proceed if there is no alternative step or it does not affect the other network. Specifically, we rely on the following rules:

\begin{itemize}
\item\textbf{Rule-M}: Control nodes are unmatched nodes in $\set V^-$, hence placing a node $v \in \set V^+$ into the matching does not affect the set of control nodes.
Therefore we apply rule-M to any leaf $v \in \set V^+$ without constraint.
On the other hand, removing a leaf $v \in \set V^-$ does affect the set of control nodes.
For example, consider the set of leaves $\set L=\{v_i^- \ \vert\ i=1,2,\ldots,k\}$, where all leaves connect to the same node $w^+$. Matching and removing $(v^-_j-w^+)$ causes leaves $\set L \setminus v^-_j $ to become isolated and hence unmatched nodes.
The unmatched nodes in $\set B$ become dominating nodes in $\set G_{\ell}$, affecting the dominating set.
To make sure that the rule-M is optimal for both the matching and dominating set problems, we only match a leaf $v^-$ if we already know that it does not have to be a dominating node, i.e., $v$ is observed and has zero out-degree in $G_{\ell}$.
\item \textbf{Rule-DS1}: Any node $v$ that becomes a source in $\set G_{\ell}$ deterministically has to become a dominating node; therefore, Rule-DS1 is applied without constraints.
\item \textbf{Rules-DS2,3}: When we apply the second or third rule, we choose a node $v$ not to be a dominating node (Fig.~\ref{fig: DSLR Rules 2} and \ref{fig: DSLR Rules 3}).
This is a useful step if only the dominating set is considered, because it is guaranteed that there is no better choice and it decimates $\set G_{\ell}$.
In the combined matching and dominating set problem, however, the step is only optimal, if $v^-$ is matched in $\set B$, i.e., it does not become a control node due to the matching condition.
\end{itemize}

We apply the above rules to $\set B$ and $\set G_{\ell}$ iteratively. If both networks are processed entirely, the obtained solution is optimal.
If we get stuck before resolving all nodes, we call the remainder of links the M-core and the DS-core of $\set B$ and $\set G_{\ell}$, respectively.
To proceed, we rely on heuristics to continue, and the solution is no longer optimal. Specifically, we match the node in $\set B$ with the least chance to be a dominating node, i.e. matching the node $v^-$ that $v$ has the least degree in $\set G_l$. If $\set B$ is empty, we select the node with the highest degree as dominating node in $\set G_{\ell}$. Then, we apply the rules to the remaining networks. We repeat this until all links are removed from both $\set B$ and $\set G_{\ell}$. An implementation of the algorithm in Python is available online~\cite{githubrepo}.

\begin{figure}[!t]
\begin{subfigure}[h]{.32\textwidth}
  \centering
  \includegraphics[width=1\linewidth]{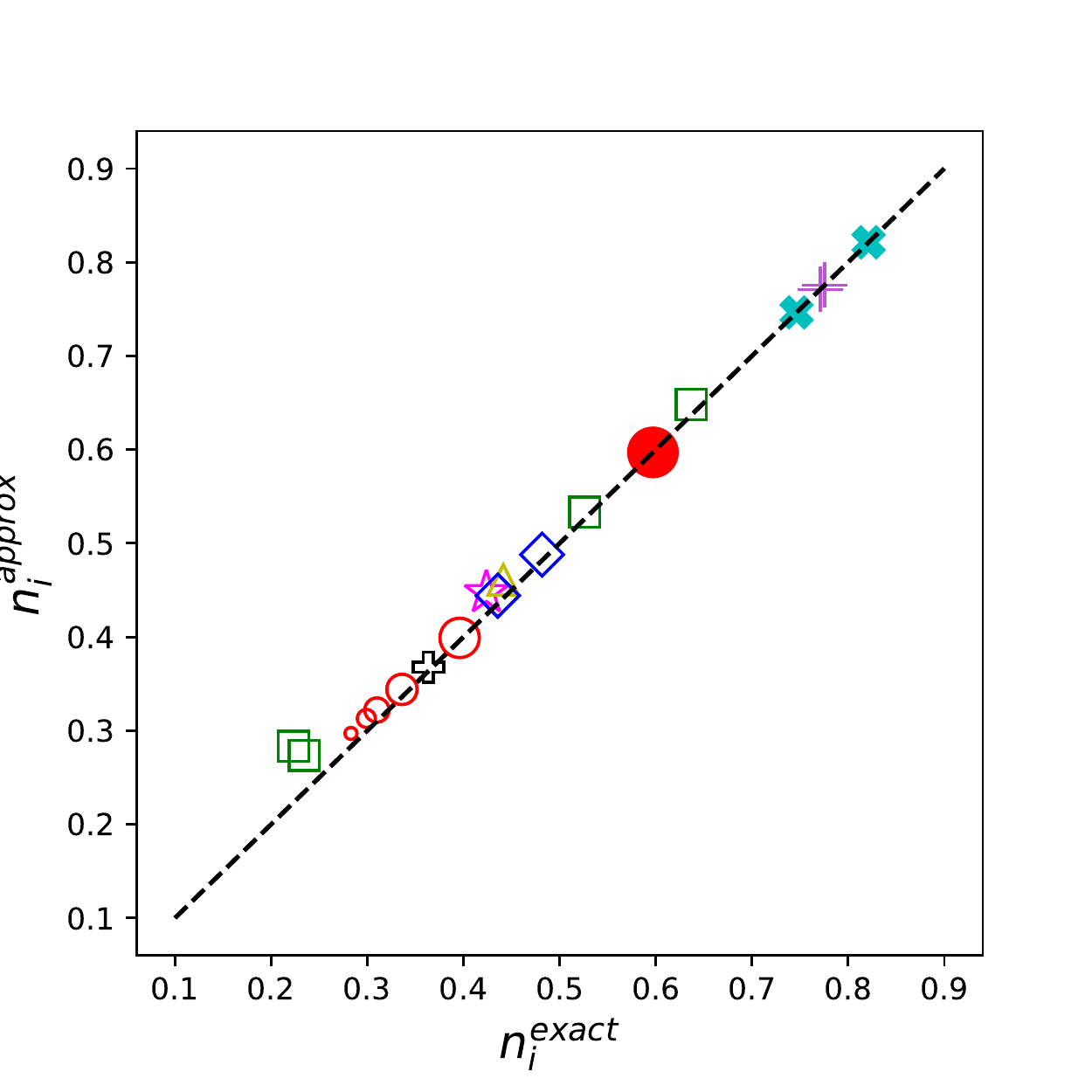}  
\caption{$\ell=1$}   
\label{fig: performance against gamma for LCC 1}
\end{subfigure}
\begin{subfigure}[h]{.32\textwidth}
  \centering
  \includegraphics[width=1\linewidth]{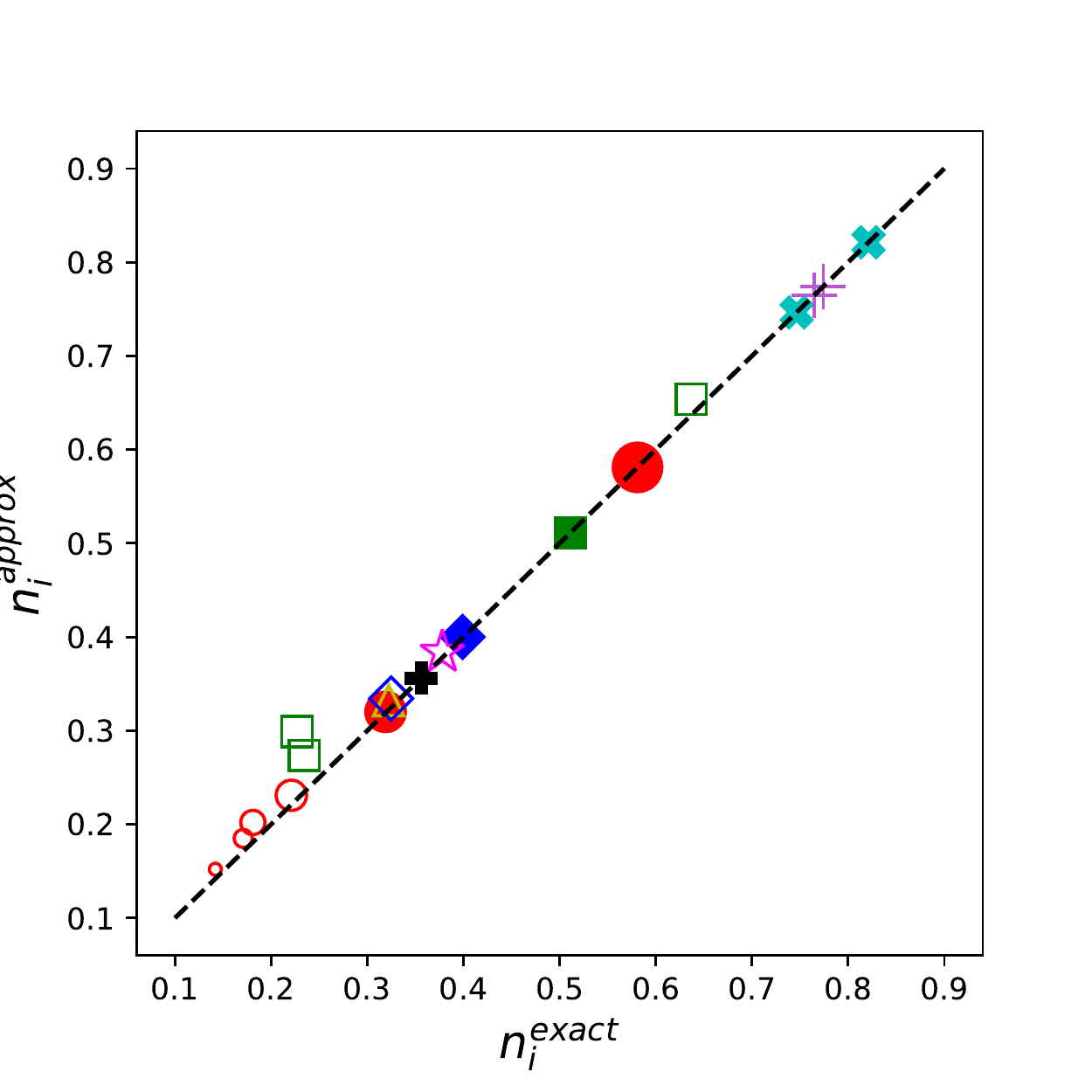}  
\caption{$\ell=2$}  
\label{fig: performance against gamma for LCC 2}
\end{subfigure}
\begin{subfigure}[h]{.34\textwidth}
  \centering
  \includegraphics[width=1\linewidth]{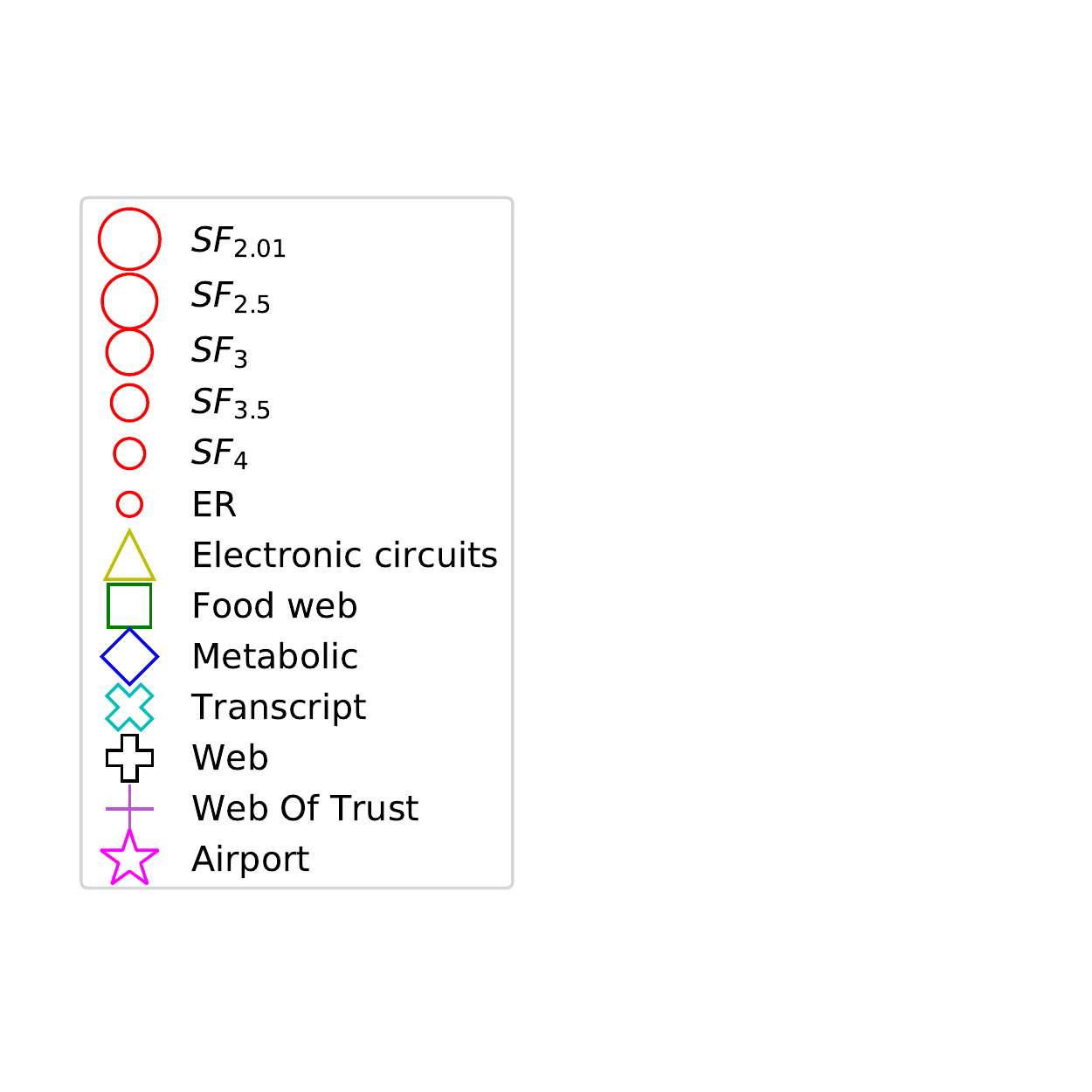}    
\label{fig: performance - legend}
\end{subfigure}
\newline
\begin{subfigure}[h]{.32\textwidth}
  \centering
  \includegraphics[width=1\linewidth]{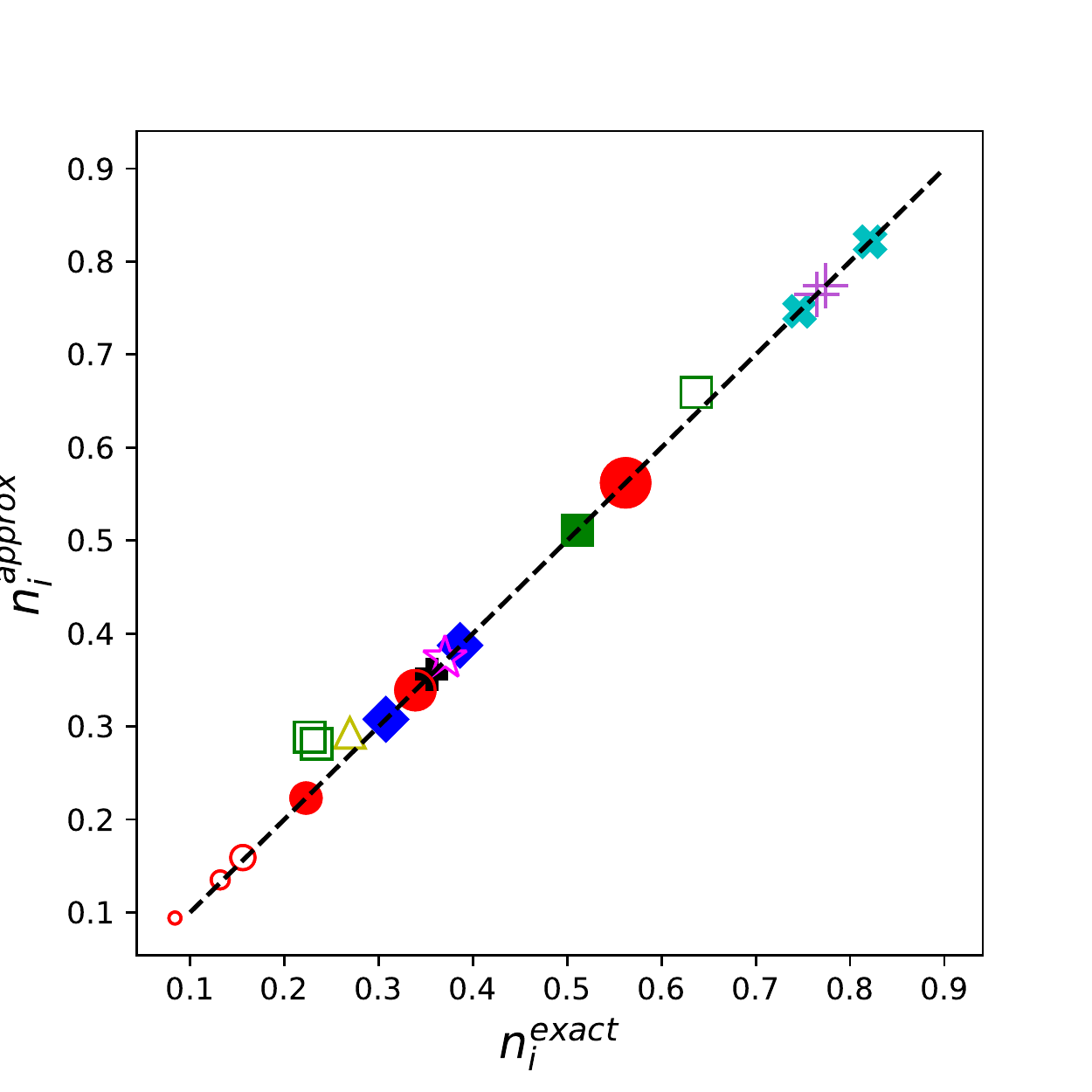}  
\caption{$\ell=3$}  
\label{fig: performance against gamma for LCC 3}
\end{subfigure}
\begin{subfigure}[h]{.32\textwidth}
  \centering
  \includegraphics[width=1\linewidth]{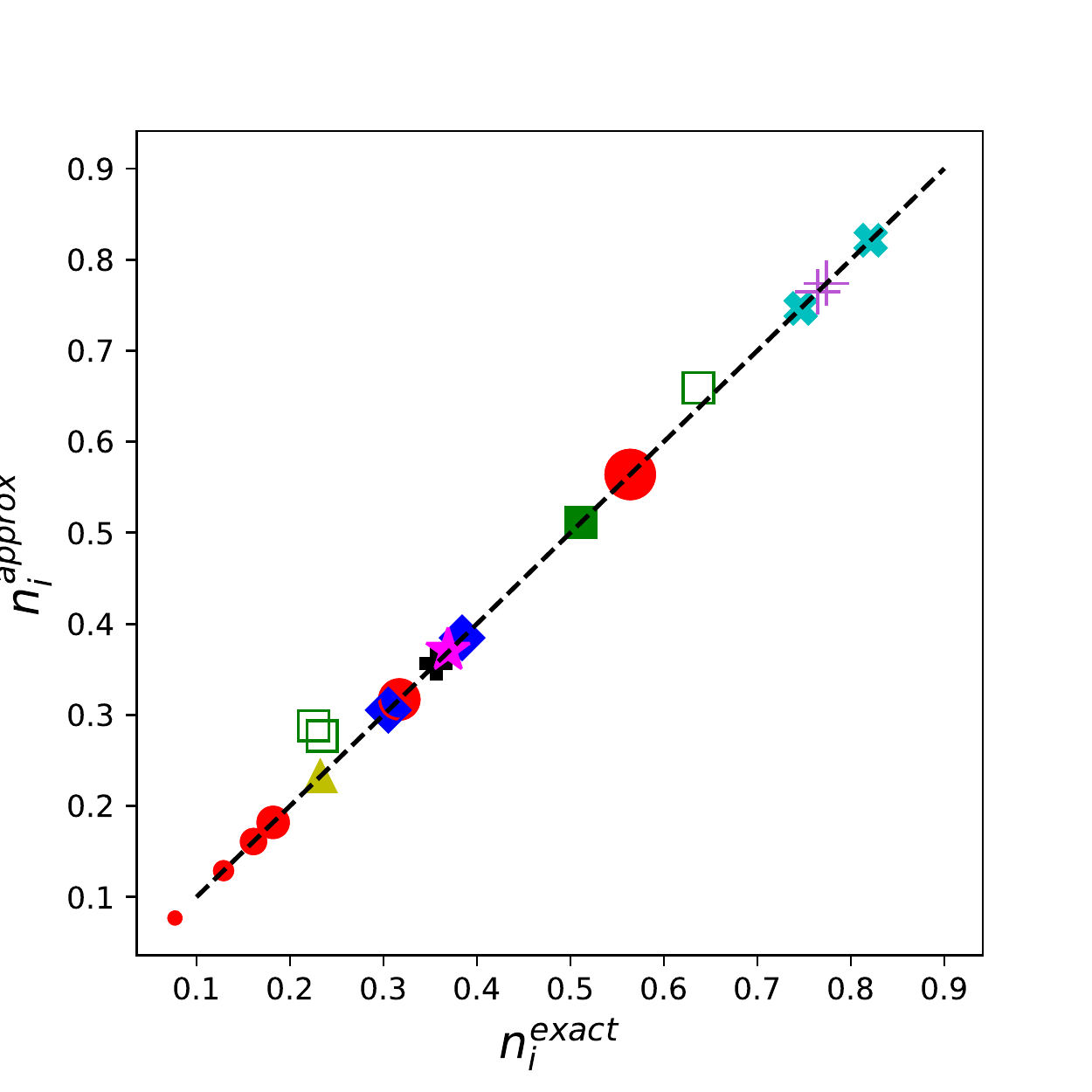}  
\caption{$\ell=4$}  
\label{fig: performance against gamma for LCC 4}
\end{subfigure}
\begin{subfigure}[h]{.34\textwidth}
  \centering
  \includegraphics[width=1.15\linewidth]{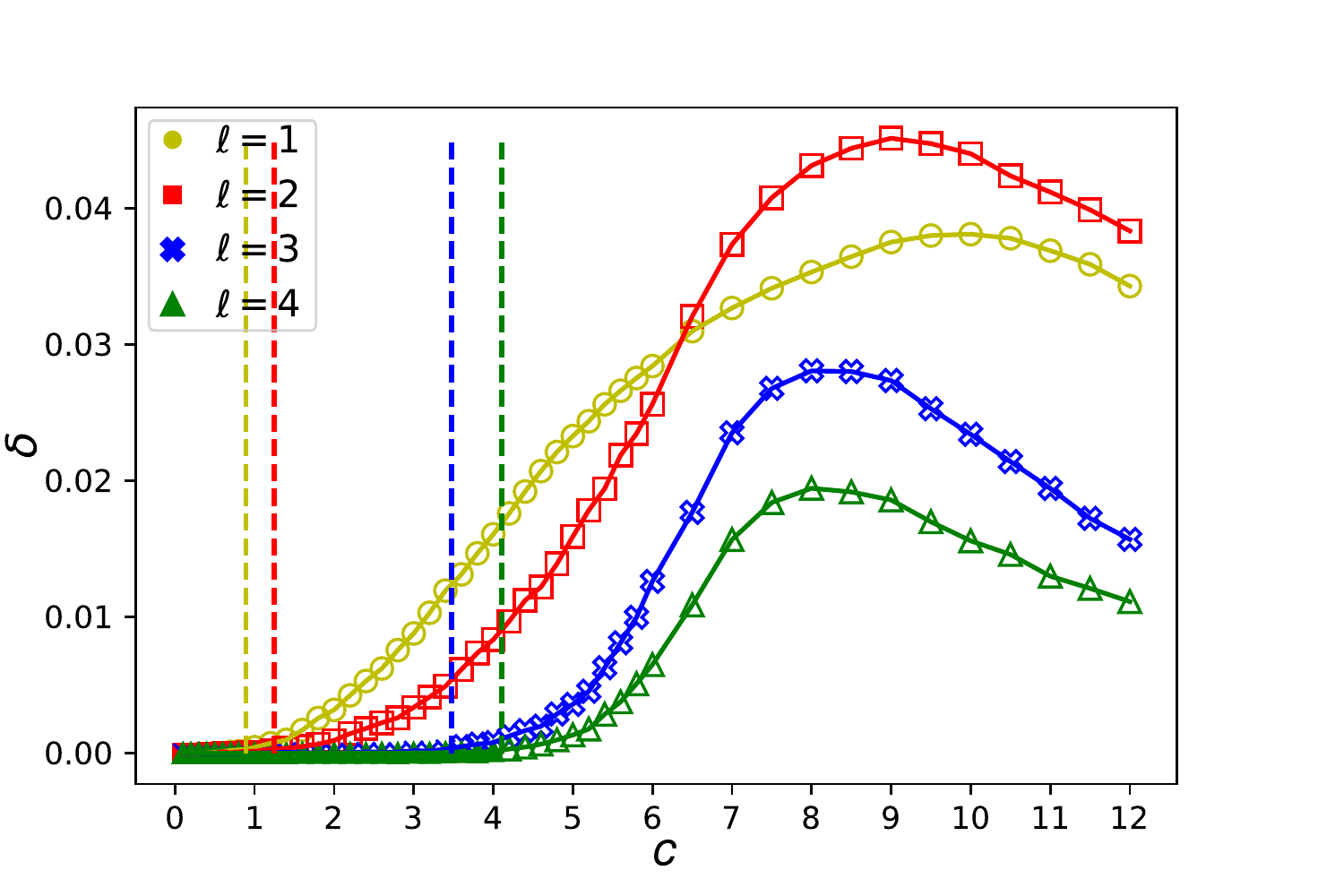}  
\caption{} 
\label{fig: performance against c}
\end{subfigure}
\caption{\textbf{The performance of the approximation.} \textbf{(a)}-\textbf{(d)} We compare the approximate fraction of input nodes $n_\T{i}^\T{approx}(\ell)=N_\T{i}^\T{approx}(\ell)/N$ to the exact solution $n_\T{i}^\T{exact}(\ell)=N_\T{i}^\T{exact}(\ell)/N$ for model and real networks with a few hundred nodes.
For networks without a core (solid markers) $n_\T{i}^\T{approx}(\ell) = n_\T{i}^\T{exact}(\ell)$, as expected. We observe a strong correlation between the approximate and exact solutions for networks with a non-zero core (hollow markers), indicating that the approximation works well for complex networks. 
Consideren networks are SF and ER model networks with $N=1000$ and average degree $c=4$. Each data point is an average of 100 independent networks. For details of the real networks see Table SI in the Supplementary Information.
\textbf{(e)}~The difference $\delta =  n^\T{approx}_i(\ell)-n^\T{exact}_i(\ell)$ as a function of the average degree shows that below the core percolation threshold the algorithm finds the exact solution, after the threshold $\delta$ increases, and the approximation becomes better for very dense networks. The data points represent the average of 1000 independently generated SF networks with $N = 300$ and $\gamma=3$.}

\label{fig: performance of heuristic algorithm}
\end{figure}

\subsection{Performance of the algorithm}

We assess the performance of our algorithm using Erd\H{o}s-R\'enyi (ER) and scale-free (SF) random networks and a collection of real networks that are small enough for the exact algorithm introduced in Sec.~\ref{sec:exact_alg}. 
Figure~(\ref{fig: performance against gamma for LCC 1})-(\ref{fig: performance against gamma for LCC 4}) compares the fraction of input nodes $n_\T i(\ell)=N_\T{i}(\ell)/N$, where $N$ is the total number of nodes in the network,  obtained using the exact and the approximate algorithm for varying $\ell$.
We find that (i)~for networks without a core the solution returned by the greedy algorithm is indeed optimal, as expected, and that (ii)~even for networks with a core the difference $\delta =  n^\T{approx}_i(\ell)-n^\T{exact}_i(\ell)$ is with less than 0.02 in real networks, except for food webs (the maximum difference 0.07 is measured for $magwet$), and the difference is less than 0.05 for synthetic networks for the parameters that we checked.
Figure~(\ref{fig: performance against c}) shows $\delta$ as a function of the average degree $c$ and the length of the LCC $\ell$ for SF networks with a given degree exponent $\gamma$.
We find that $\delta$ is the highest for intermediate $c$ values: $\delta$ becomes non-zero at the point when the core forms and initially increases as the size of the core increases. 
However, $\delta$ starts to drop for large $c$, since adding links to a network increases the maximum matching and reduces the minimum dominating set, making the LCC-constraint minimum input problem easier to approximate.
Increasing the length of the LCC $\ell$ increases the density of $\set G_\ell$; we, therefore, see a similar pattern: the error $\delta$ is the highest for $\ell=2$, and further increasing $\ell$  makes the problem easier to approximate.  
 
The algorithm finds the exact solution if there is no core in the network; therefore, another way to understand the performance of the approximation is to investigate how network properties affect the formation of the core.
For this we numerically measure the size of the core for ER and SF random networks as a function of the average degree.
For the maximum matching problem by itself the M-core emerges through a second order phase transition at a critical average degree $c_\T{M}^*$, for ER random networks $c_\T{M,ER}^*=e$, and in SF random networks the formation of the core is typically delayed, i.e., $c_\T{M,SF}^*>c_\T{M,ER}^*$.
Figure~(\ref{fig: Cores}) shows that for $\ell=1$ the core is already present for any $c>0$, and we only observe the phase transition for $\ell \geq 3$ for ER, and $\ell \geq 2$ for SF networks, further increasing $\ell$ delays the emergence of the core, supporting our previous observation that for $\ell>2$ increasing $\ell$ improves the performance of the approximation.
We also observe that the emergence of the core happens at larger average degree $c$ for SF networks than for ER network, i.e., $c_\T{SF}^*(\ell)\leq c_\T{ER}^*(\ell)$, providing evidence that heterogeneous networks are easier to approximate.
For details about the core formation in real networks, see Table~SII in the Supplementary Information.

Overall, we find that the approximation algorithm estimates $N_\T{i}(\ell)$ within 5\% error for the networks that we tested.
The performance of the algorithm is enhanced by degree heterogeneity, an ubiquitous feature of complex networks, and high $\ell$, and it is inhibited by the emergence of the core.
While the exact algorithm is limited to networks with a few hundred links, in the next section we apply the leaf-removal approximation to networks with up to five million links.

\begin{figure}[h]
\begin{subfigure}[h]{.49\textwidth}
  \centering
  \includegraphics[width=.8\linewidth]{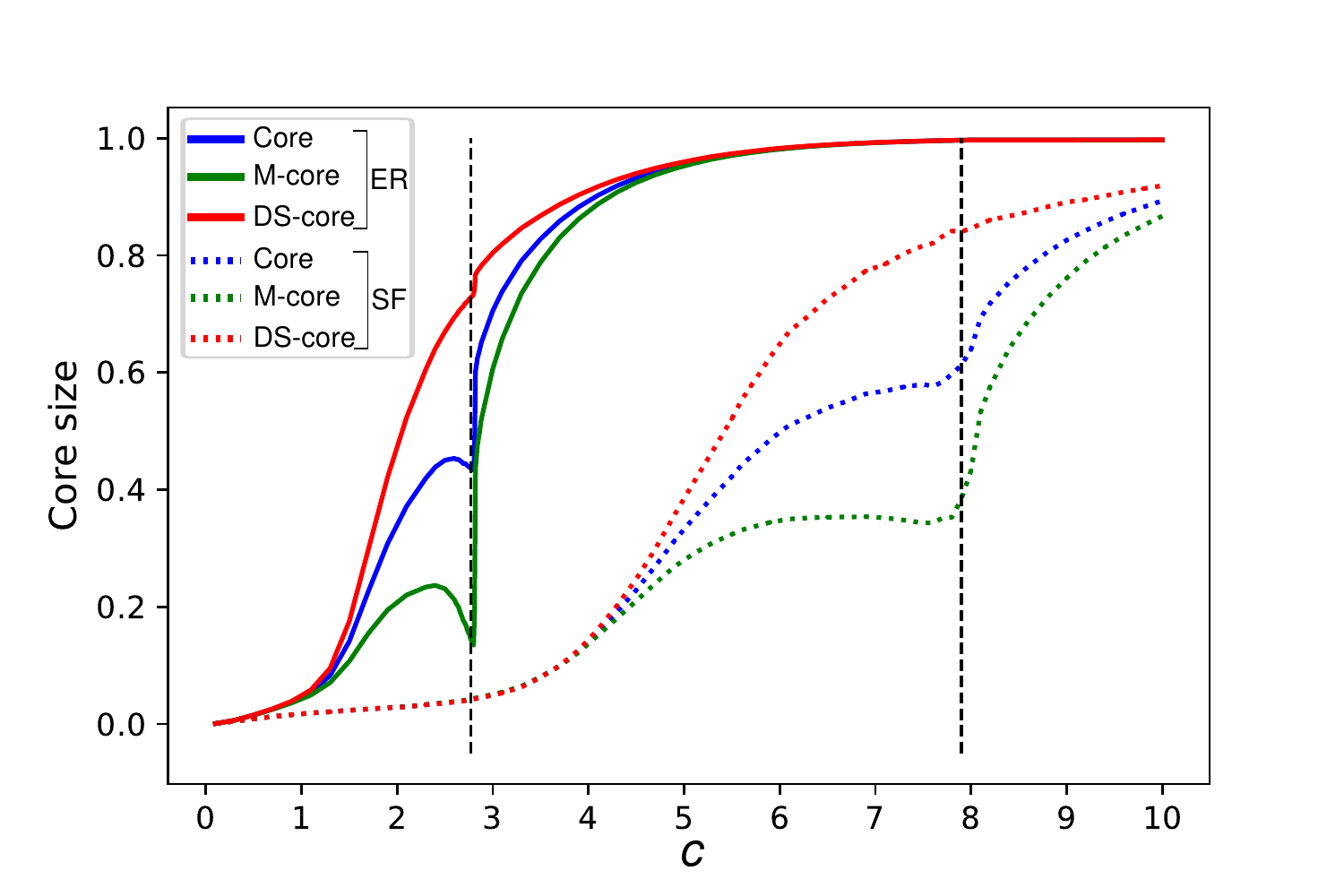}  
\caption{$\ell=1$}  
\label{fig: Core in LCC = 1}
\end{subfigure}
\begin{subfigure}[h]{.49\textwidth}
  \centering
  \includegraphics[width=.8\linewidth]{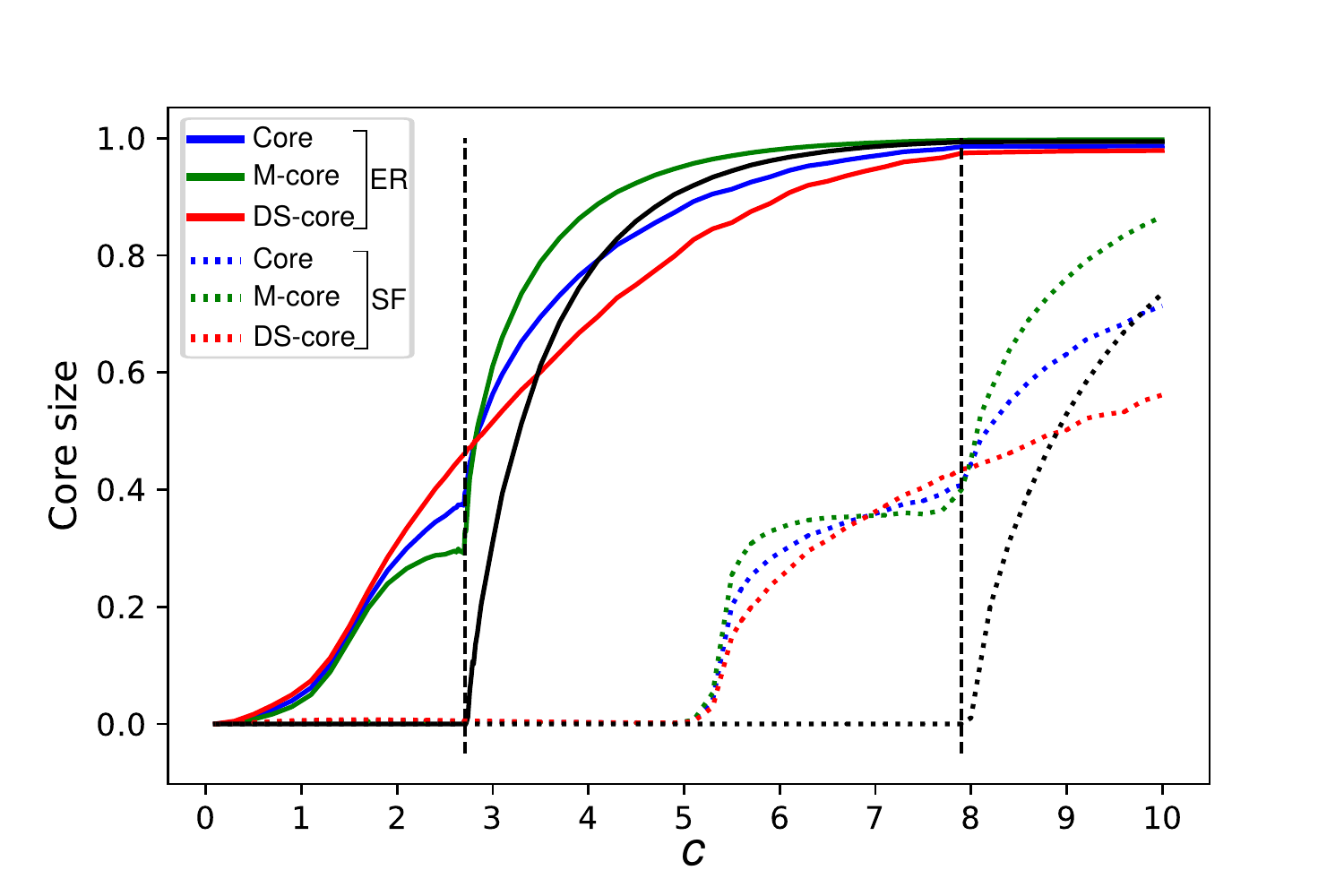}  
\caption{$\ell=2$}  
\label{fig: Core in LCC = 2}
\end{subfigure}
\newline
\begin{subfigure}[h]{.49\textwidth}
  \centering
  \includegraphics[width=.8\linewidth]{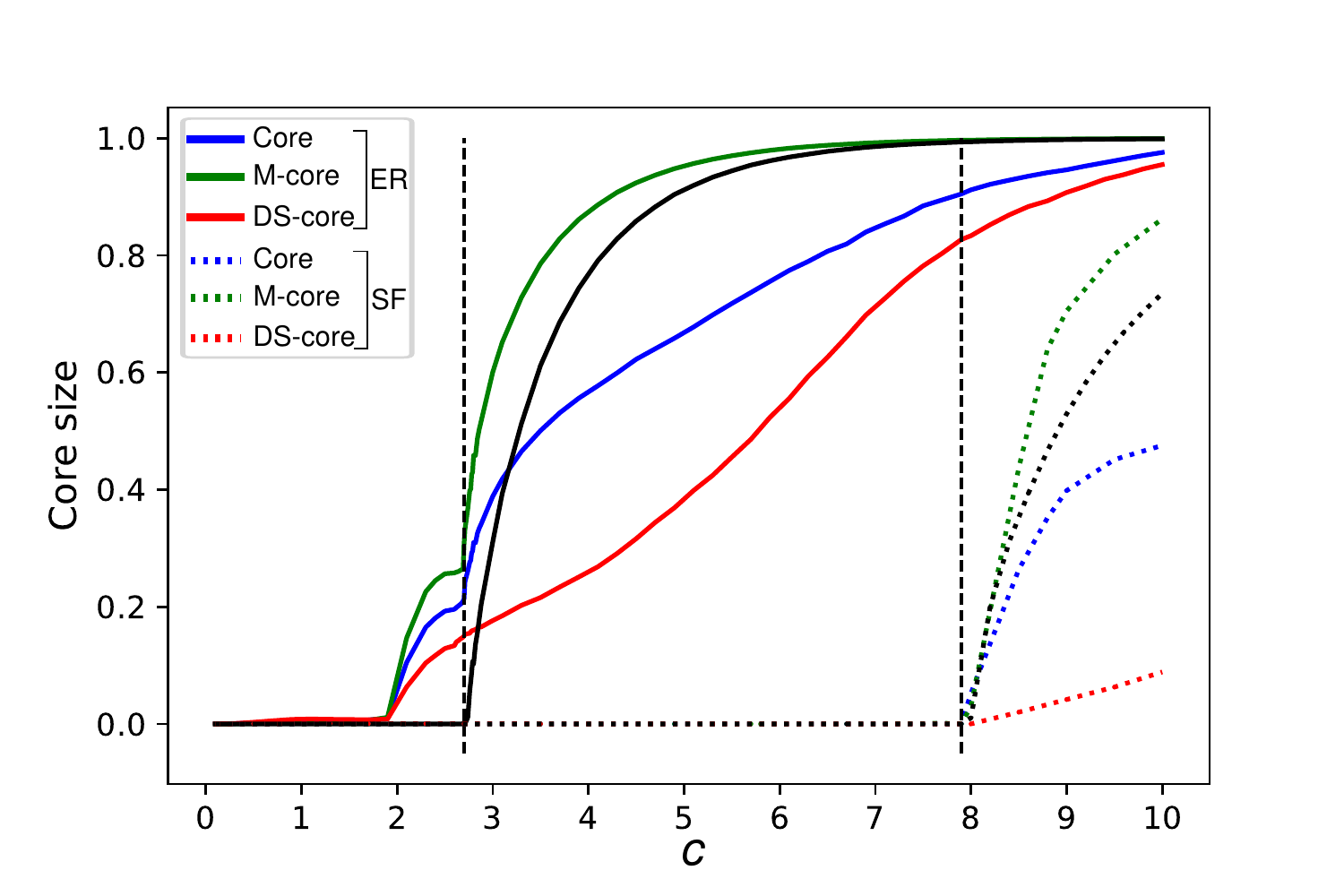}  
\caption{$\ell=3$}  
\label{fig: Core in LCC = 3}
\end{subfigure}
\begin{subfigure}[h]{.49\textwidth}
  \centering
  \includegraphics[width=.8\linewidth]{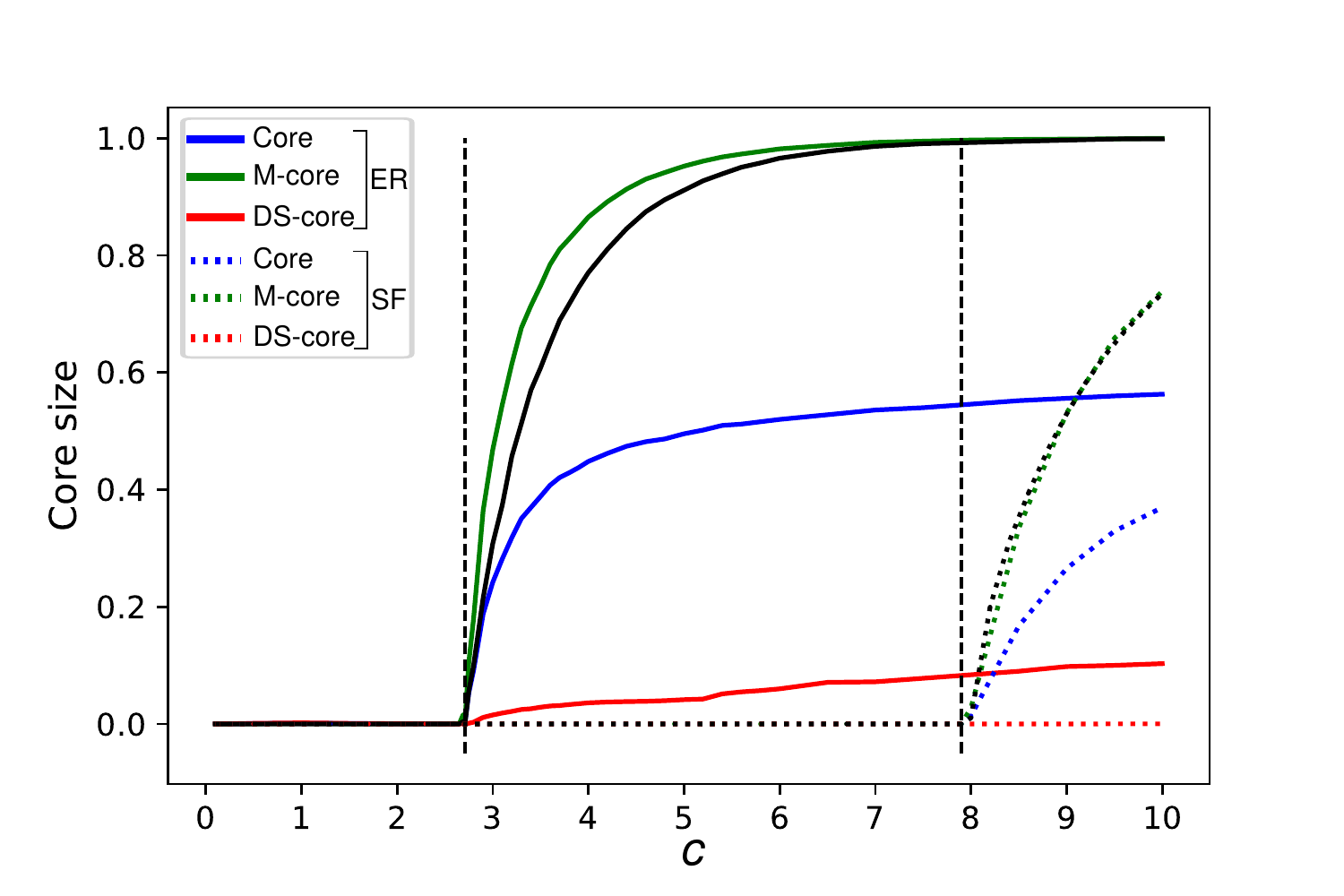}  
\caption{$\ell=4$}  
\label{fig: Core in LCC = 4}
\end{subfigure}
\caption{\textbf{Core percolation in ER and SF networks}. We measure the relative size of core in the network (blue), which is the average of the M-core (green) and the DS-core (red), for ER and SF random networks as a function of the average degree. We find that for $\ell=1$ the core is present for any $c>0$, and the core appears at a finite critical $c^*$ for $\ell \geq 3$ for ER and $\ell \geq 2$ for SF networks.
For the maximum matching problem by itself the M-core emerges at average degree $c_\T{M}^*$ (vertical dashed line), this critical point always effects the core, and numerical evidence suggests that $c^*=c_\T{M}^*$ for $\ell \geq 3$ for ER and $\ell \geq 2$ for SF networks.
The core size is calculated for networks with $N=10^6$ and the SF networks have degree exponent $\gamma=2.5$. The data points are an average of 100 independent networks, and the error bars are smaller than the line width.}
\label{fig: Cores}
\end{figure}

\section{Results} \label{sec: result}

In the previous section we developed an algorithm to approximately solve the LCC-constrained control problem for large complex networks.
Relying on this algorithm we now systematically explore how network structure affects the minimum number of inputs $N_{\T i}(\ell)$ and the cost of enforcing $l_\T{LCC} \leq \ell$.

\subsection{Number of input nodes}

First, we measure $N_{\T i}(\ell)$ for Erd\H os-R\'enyi (ER) and scale-free (SF) model networks to explore the role of the degree distribution.
The SF model we use is a variant of the hidden parameter model, which allows us to control both the average degree $c$ and the degree exponent $\gamma$ characterizing the tail of the degree distribution~\cite{dorogovtsev2008critical}.
Figure~(\ref{fig: Lower-Approx-Upper}) shows $N_{\T i}({\ell})$ as a function of $c$, $\gamma$, and $\ell$, to help interpret the results we also include the lower bound~(\ref{eq:LB}) and upper bound~(\ref{eq:UB}) .
We find that dense networks are easier to control than sparse networks for any $\ell$ (Figs.~(\ref{fig: 3points against c - LCC1})-(\ref{fig: 3points against c - LCC3})), similarly to previous results not taking into account the LCC~\cite{liu}.
We also observe that the upper and lower bound converge towards each other as the average degree $c$ increases.
To explain this, note that the maximum matching rapidly grows with increasing $c$, hence the number of unmatched nodes rapidly decreases.
For example, in ER networks, the number of unmatched nodes drops exponentially as a function of $c$~\cite{liu}.
The minimum dominating set also decreases with increasing $c$, but the decay is significantly slower.
For example, for ER networks, the minimum dominating set decreases slower than $c^{-1}$~\cite{habibulla2015directed}.
It means that for large $c$, both the lower and the upper bound approach $N_\T{DS}$ (the size of a minimum dominating set in $\set G_\ell$ without considering the matching), and matching has a diminishing effect on $N_{\T i}(\ell)$.

The role of degree heterogeneity is also consistent with previous results: Figs.~(\ref{fig: 3points against gamma - LCC1})-(\ref{fig: 3points against gamma - LCC3}) show that for any $\ell$ homogeneous networks require fewer inputs $N_{\T i}(\ell)$ than heterogeneous networks with the same average degree $c$.
We find that the upper and lower bounds converge towards each other as the degree exponent $\gamma$ decreases, which corresponds to increasing degree heterogeneity.
Both the number of hubs and low-degree nodes increase with increasing degree heterogeneity, which in turn decreases the size of the maximum matching and increases $N_\T{M}$ (the number of unmatched nodes in $\set B$ without considering the dominating set).
The effect of $\gamma$ on the minimum dominating set is more complex: the presence of hubs generally reduces $N_\T{DS}$, while the presence of sources $N_\T s$, i.e., nodes with zero in-degree, increases $N_\T{DS}$. 
The net effect is that $N_\T{DS}$ converges to $N_\T s$ as $\gamma$ decreases; source nodes, however, are also always unmatched nodes, hence in the limit of very heterogeneous networks, the LCC problem is determined by the maximum matching of $\set B$.

Figure~(\ref{fig: 3points against LCC}) shows the affect of the LCC $\ell$: increasing $\ell$ initially lowers $N_{\T i}({\ell})$, however once $\ell$ exceeds the diameter $D$ of the network, $N_{\T i}({\ell})$ becomes constant.
For large $\ell$, we again observe that the lower and the upper bounds approach each other since with increasing $\ell$, the accessibility graph $\set G_\ell$ becomes more densely connected and the minimum input problem becomes largely determined by the maximum matching problem.

So far we studied the role of the degree distribution using network models where we fix the expected degree of each node, but are otherwise random.
To reveal the role of higher-order structural properties, such as degree correlations or community structure~\cite{xiang2015analysis,mayo2015long,newman2012communities,posfai2013effect}, we first measure $N_\T{i}(\ell)$ for a collection of real networks,for details see Table~SI in the Supplementary Information.
We then randomize the networks while preserving their degree distribution but otherwise randomly rewiring their links, thus removing any higher-order structure.
Finally we measure $N_\T{i}^\T{rand}(\ell)$, the number of input signals needed to control these randomized networks.
Figure~(\ref{fig: degree-preserving randomization}) compares $N_\T{i}(\ell)$ to its randomized counterpart $N_\T{i}^\T{rand}(\ell)$:
If the degree sequence of a network would completely determine the number of control signals, the two values would be equal, any difference is due to the structural features removed by randomization.
We observe a strong correlation between $N_\T{i}(\ell)$ and $N_\T{i}^\T{rand}(\ell)$, indicating that degree distribution indeed plays an important role in determining the controllability of real complex networks.
Some markers, however, deviate from the diagonal, food webs being the strongest outliers. These food webs are relatively small and strongly heterogeneous, and their control characteristics are often observed to deviate from what is expected based on their degree distribution alone~\cite{posfai2013effect,ghasemi2020diversity}.

\begin{figure}[!h]
\begin{subfigure}[h]{.33\textwidth}
  \centering
  \includegraphics[width=1.1\linewidth]{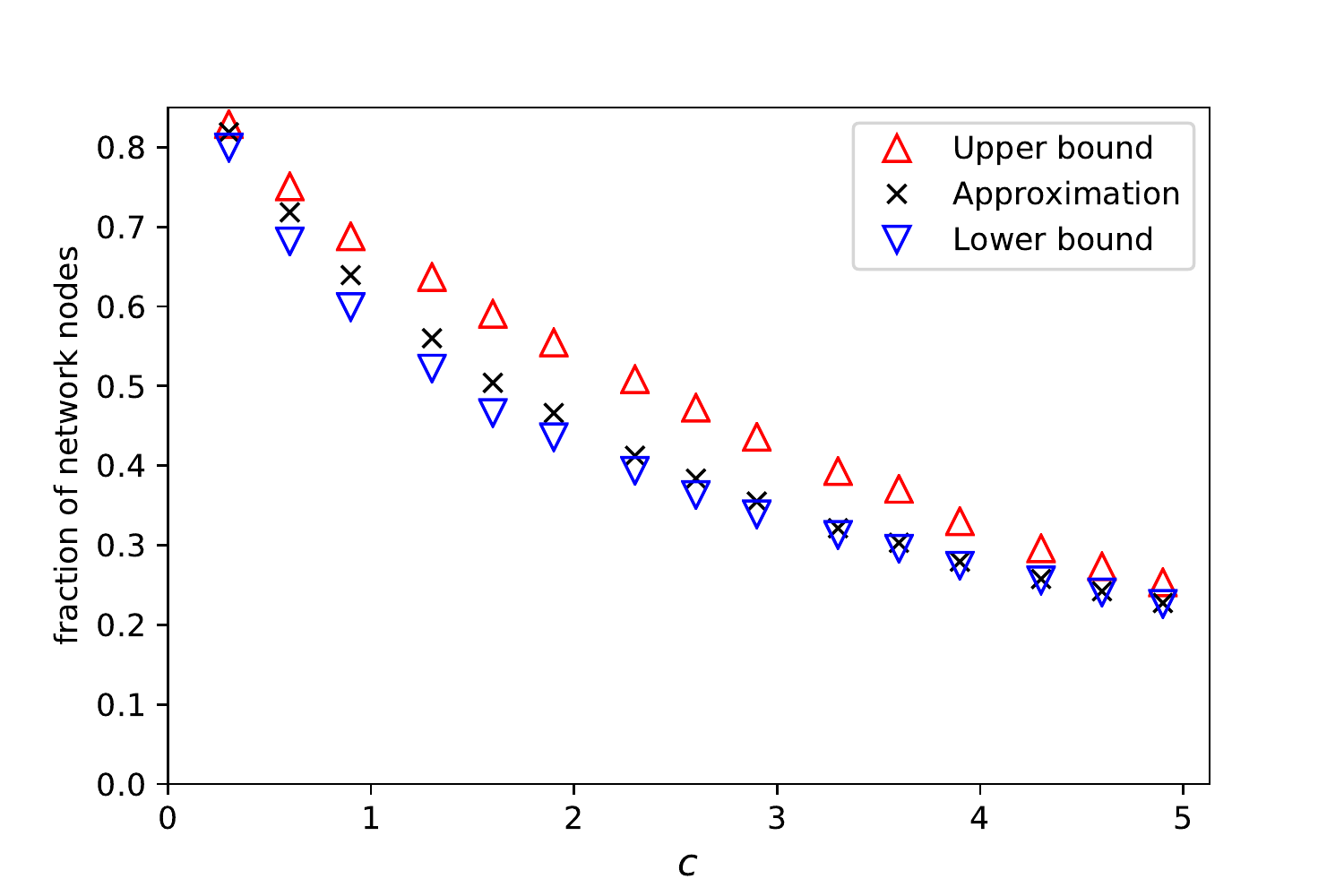}  
\caption{$\ell=1$}  
\label{fig: 3points against c - LCC1}
\end{subfigure}
\begin{subfigure}[h]{.33\textwidth}
  \centering
  \includegraphics[width=1.1\linewidth]{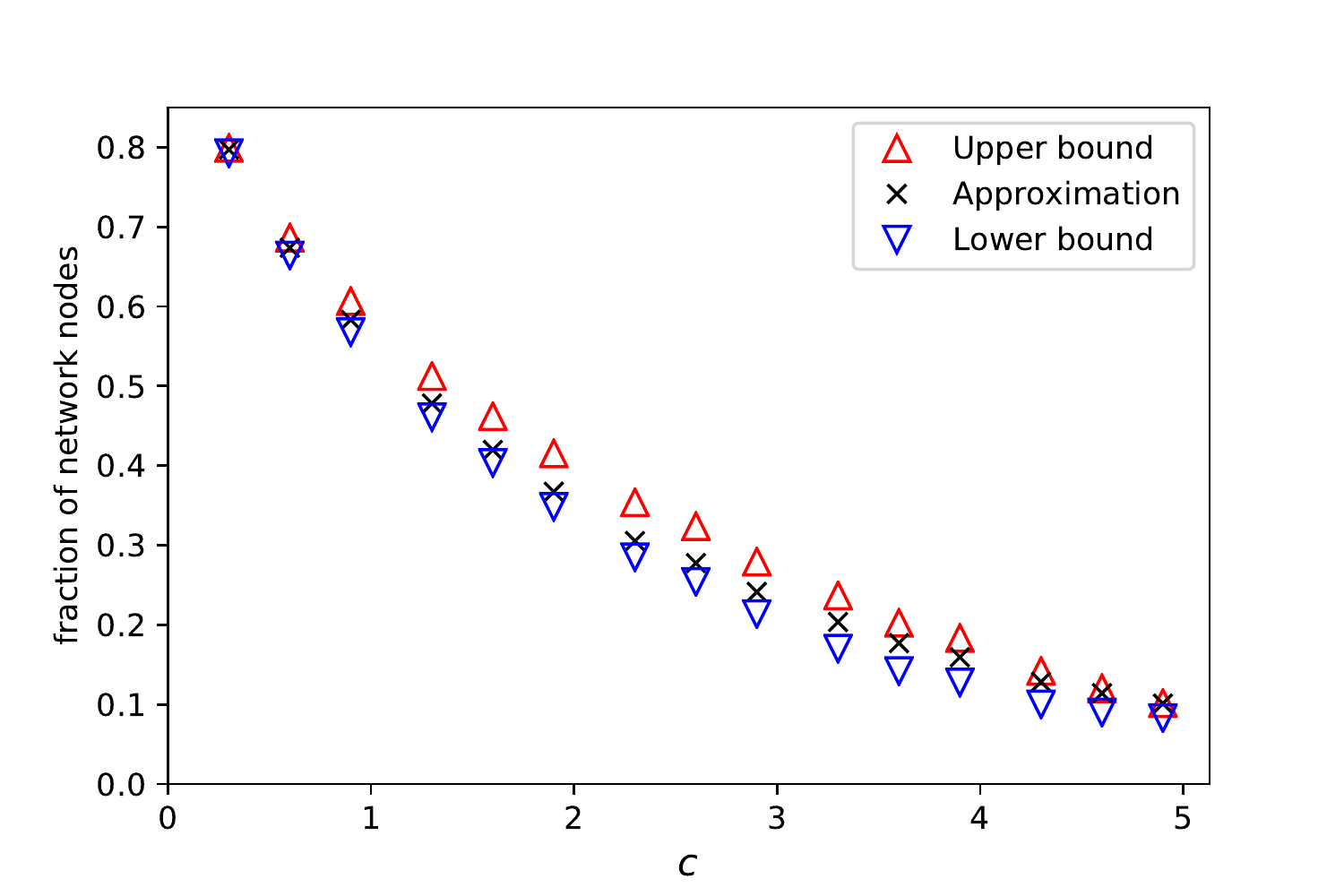}  
\caption{$\ell=2$}  
\label{fig: 3points against c - LCC2}
\end{subfigure}
\begin{subfigure}[h]{.33\textwidth}
  \centering
  \includegraphics[width=1.1\linewidth]{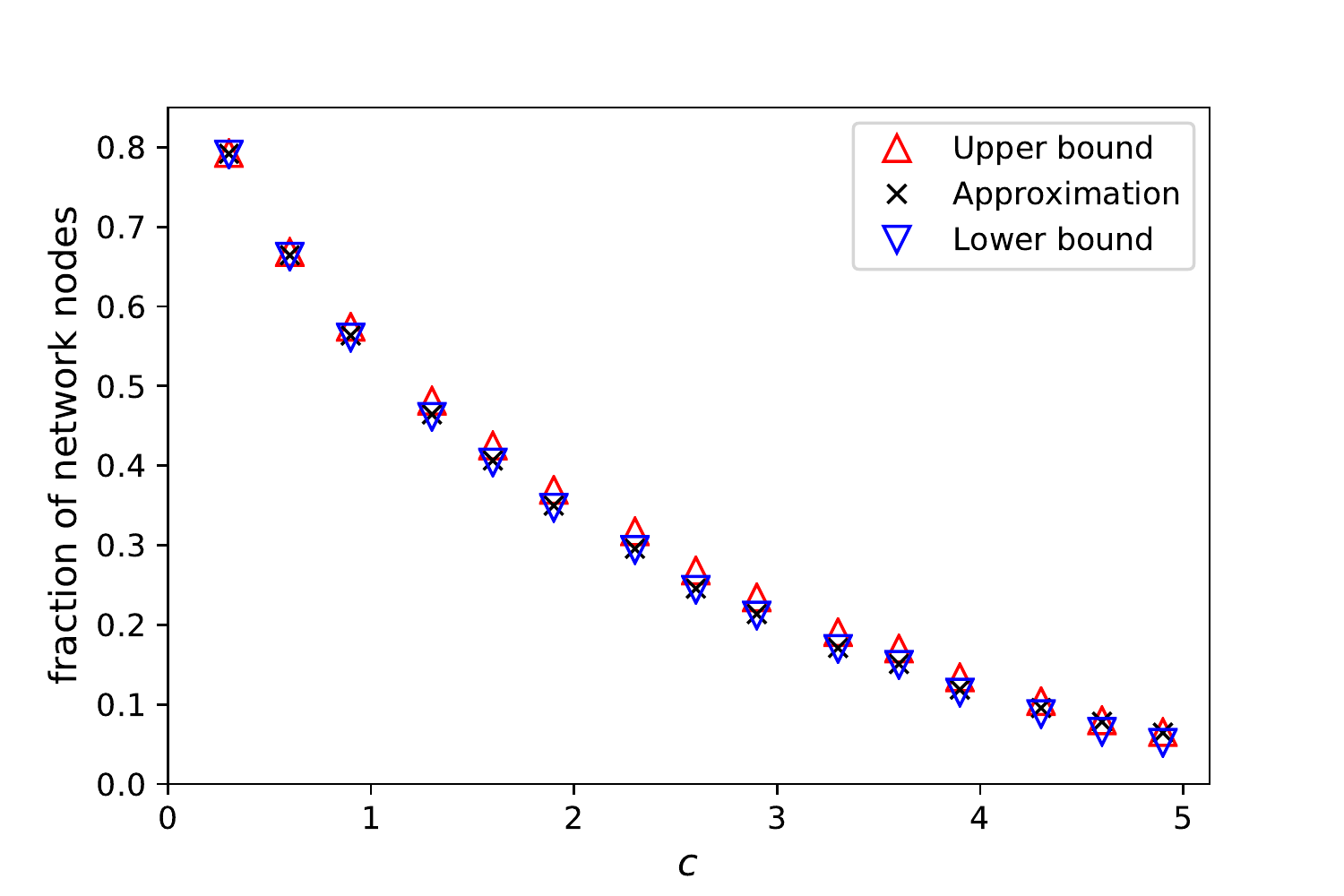}  
\caption{$\ell=3$}  
\label{fig: 3points against c - LCC3}
\end{subfigure}
\newline
\begin{subfigure}[h]{.33\textwidth}
  \centering
  \includegraphics[width=1.1\linewidth]{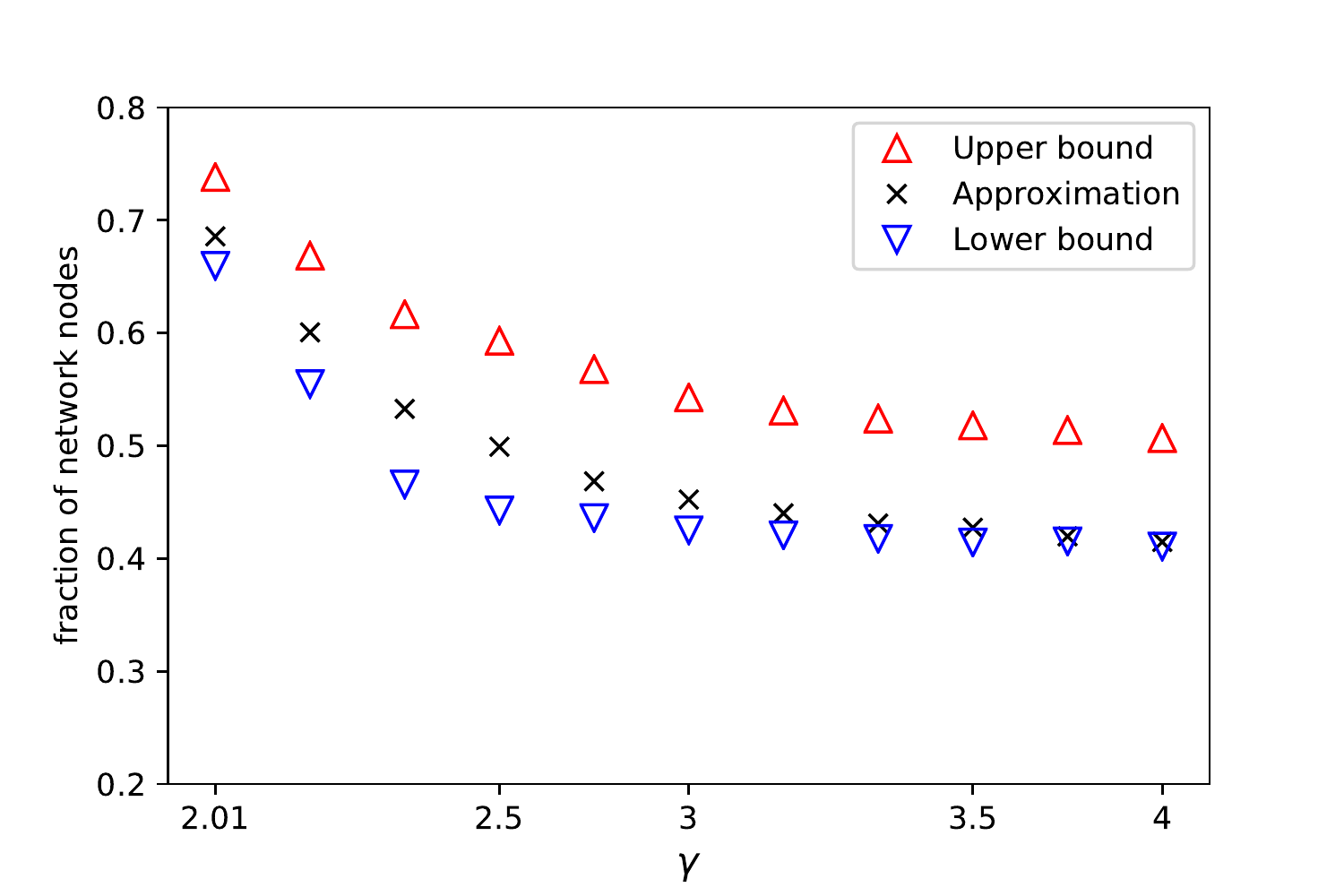}  
\caption{$\ell=1$}  
\label{fig: 3points against gamma - LCC1}
\end{subfigure}
\begin{subfigure}[h]{.33\textwidth}
  \centering
  \includegraphics[width=1.1\linewidth]{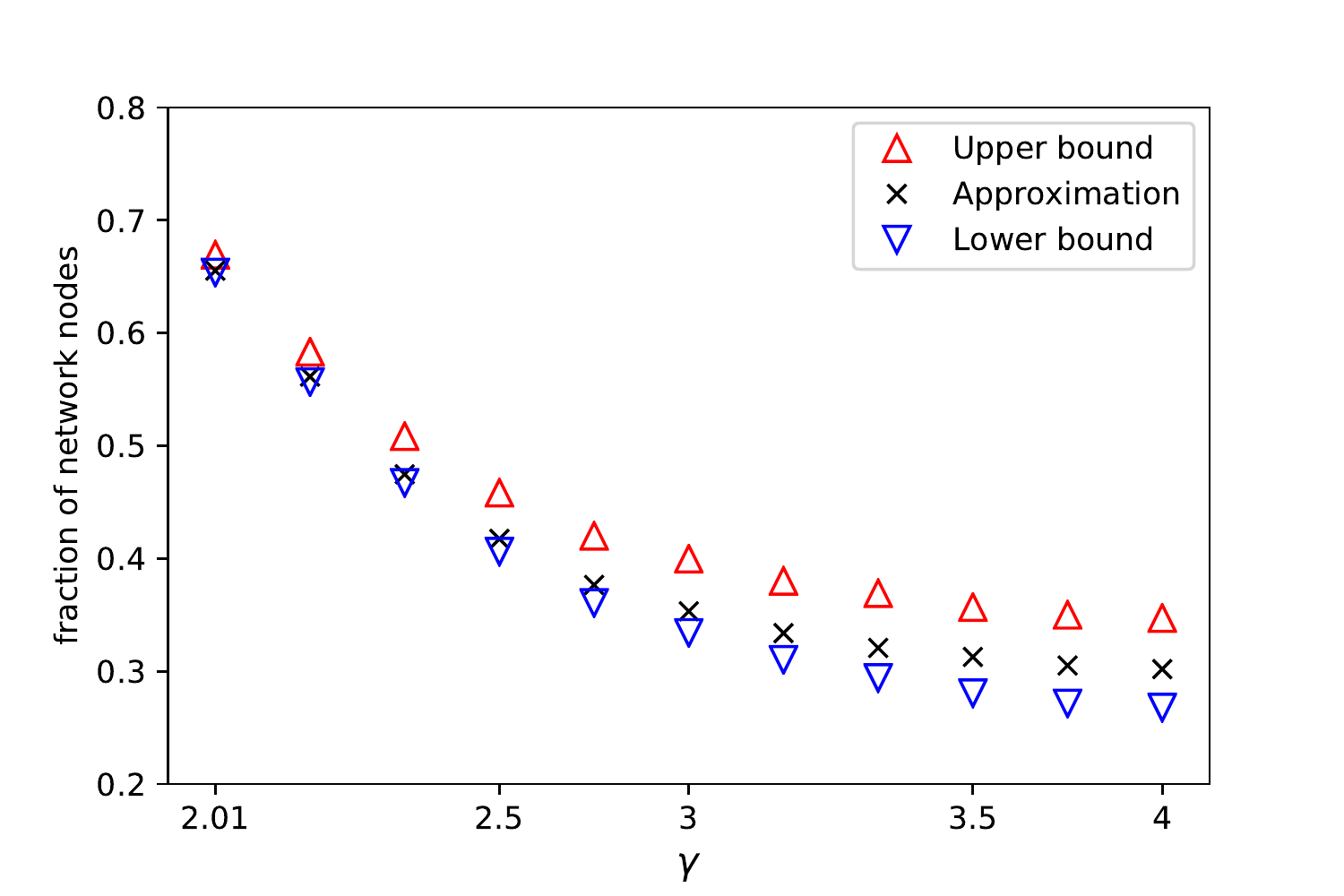}  
\caption{$\ell=2$}  
\label{fig: 3points against gamma - LCC2}
\end{subfigure}
\begin{subfigure}[h]{.33\textwidth}
  \centering
  \includegraphics[width=1.1\linewidth]{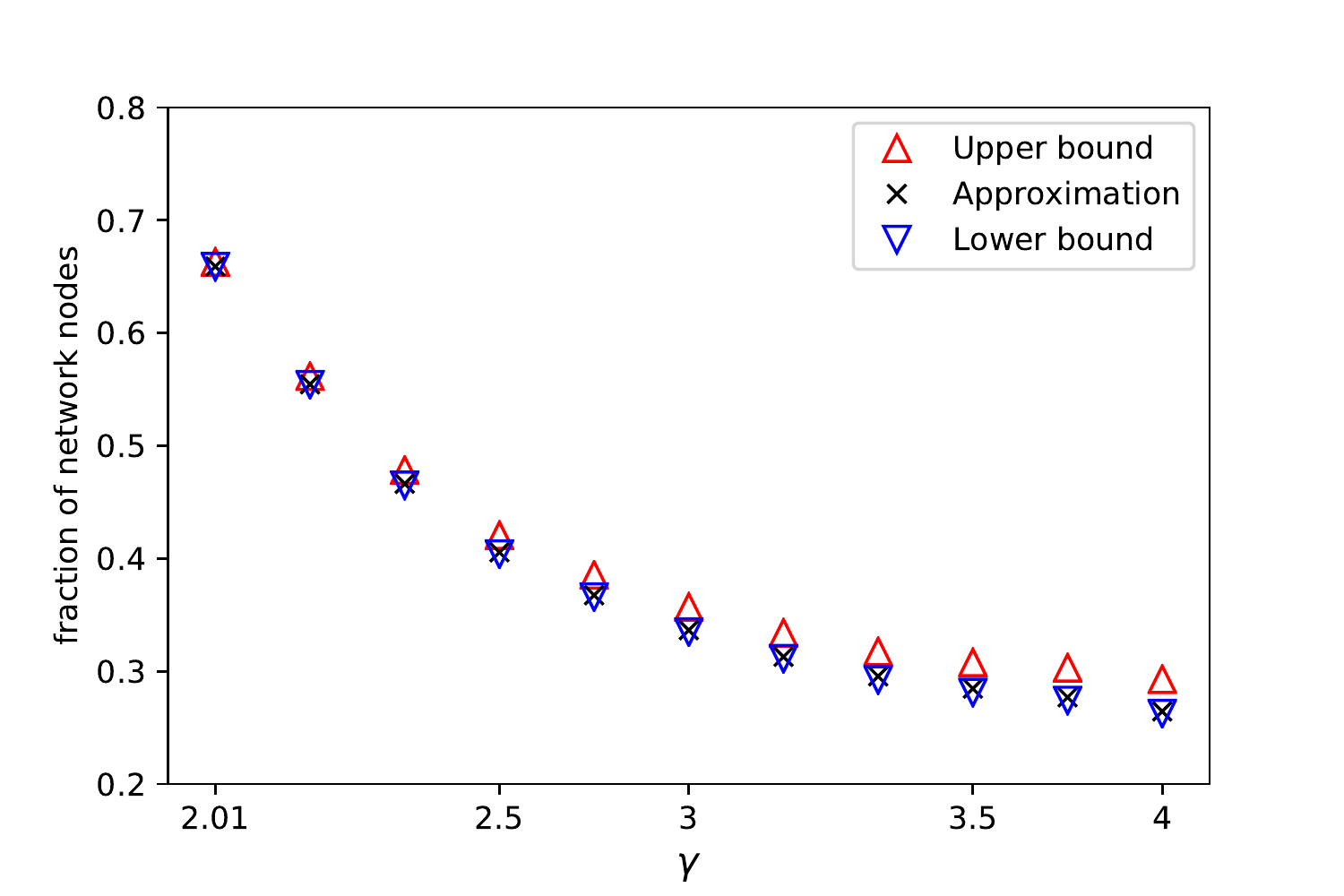}  
\caption{$\ell=3$}  
\label{fig: 3points against gamma - LCC3}
\end{subfigure}
\newline
\begin{subfigure}[h]{.33\textwidth}
  \centering
  \includegraphics[width=1.1\linewidth]{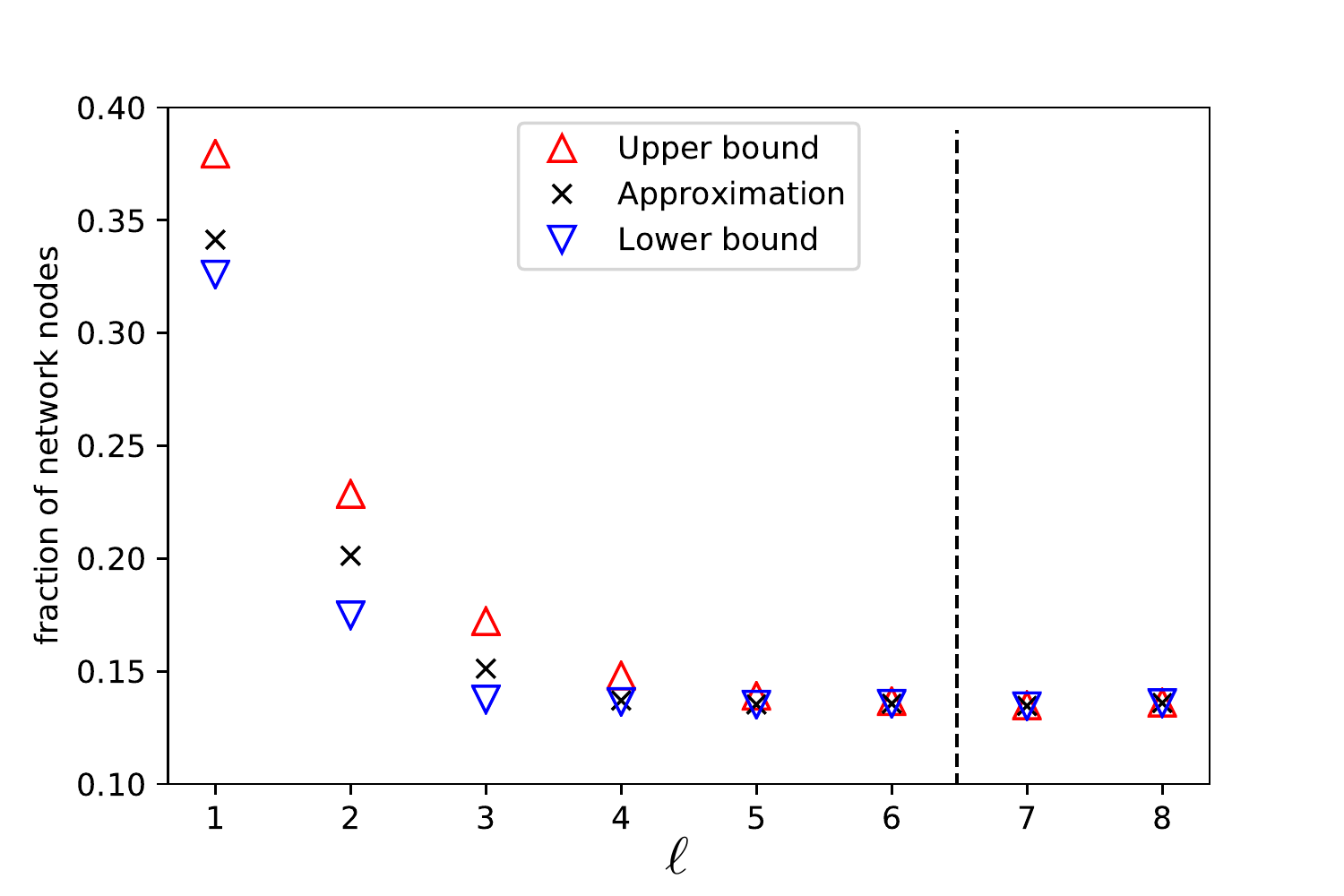} 
\caption{}   
\label{fig: 3points against LCC}
\end{subfigure}
\caption{\textbf{Minimum number of inputs $N_\textrm{i}(\ell)$ -- model networks.} 
\textbf{(a)}-\textbf{(c)}~Increasing average degree $c$ decreases $N_\textrm{i}(\ell)$ for any $\ell>0$, similarly to previous results obtained without considering the LCC.
The lower bound~(\ref{eq:LB}) and upper bound~(\ref{eq:UB}) converge towards each other in the dense network limit, since $N_\T{M}$ decreases rapidly and both bounds are determined by $N_\T{DS}$.
Plots show results obtained for SF networks with $\gamma=3$ and $N=10^4$; we observe similar behavior for other $\gamma$ values and ER networks.
\textbf{(d)}-\textbf{(f)}~Measuring $N_\textrm{i}(\ell)$ for SF networks as a function of the degree exponent $\gamma$ shows that homogeneous networks are easier to control for any $\ell>0$, consistent with previous results.
Plots show SF networks with $c=2$ and $N=10^4$; we observe similar behavior for other $c$ values.
\textbf{(g)}~Increasing $\ell$, the maximum allowed length of the LCC, decreases $N_\textrm{i}(\ell)$, since the accessibility graph $\set G_{\ell}$  becomes denser for increasing $\ell$; and therefore the accessibility condition becomes easier to satisfy. 
This also means that the lower and upper bounds are both determined by $N_\T{M}$.
The horizontal dashed line indicates the average shortest path in the network, above $N_\T{i}(\ell)$ becomes independent of $\ell$.
The plot shows ER networks with $N=10^3$ and $c=2.5$, we observe similar behavior for other $c$ values and SF networks.
Data points in the figure indicate the average of 100 independent networks, error bars are smaller than the markers.
}
\label{fig: Lower-Approx-Upper}
\end{figure}

\begin{figure}[!h]
\begin{subfigure}[h]{.33\textwidth}
  \centering
  \includegraphics[width=1\linewidth]{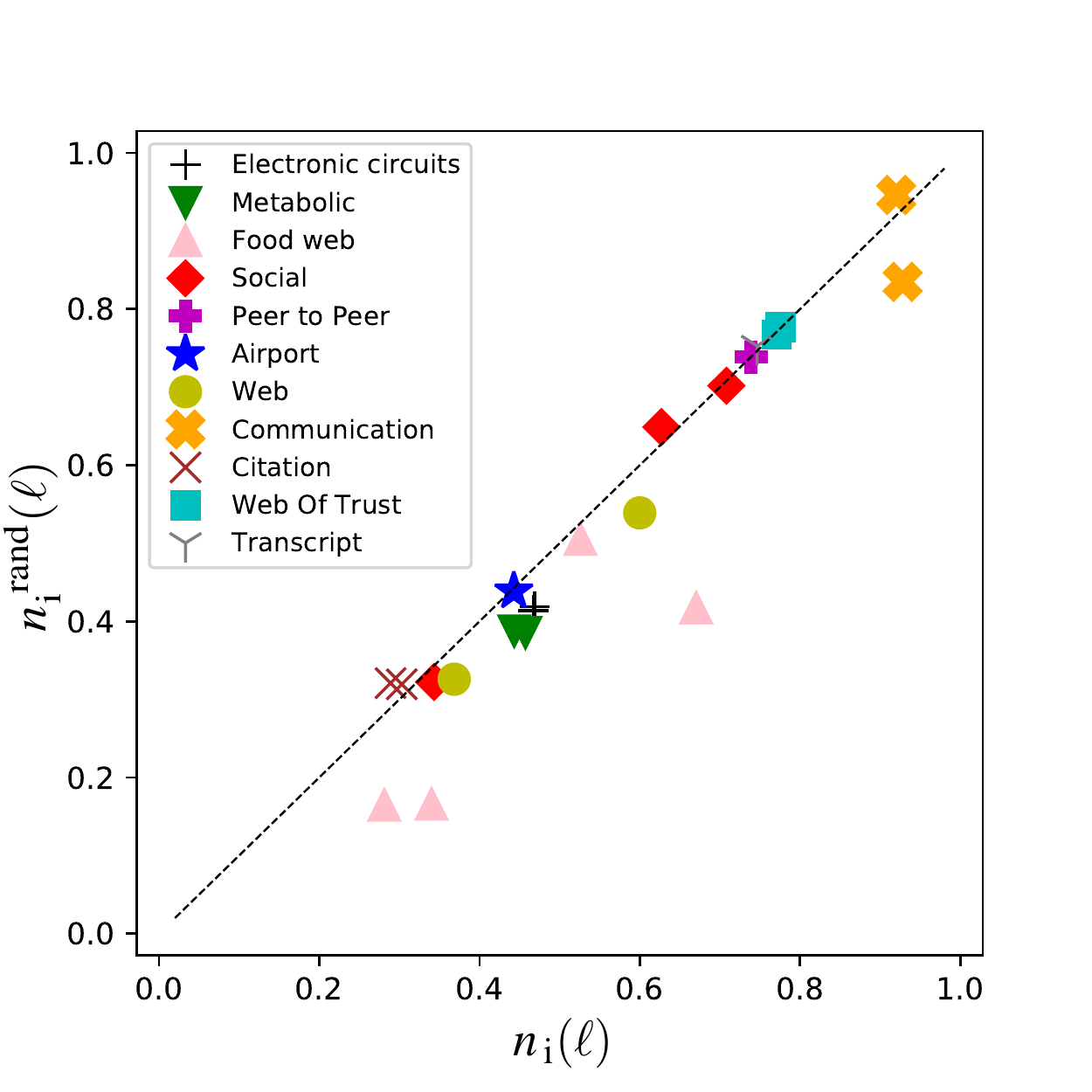}  
\caption{$\ell=1$}  
\label{fig: degree-preserve randomization LCC=1}
\end{subfigure}
\begin{subfigure}[h]{.33\textwidth}
  \centering
  \includegraphics[width=1\linewidth]{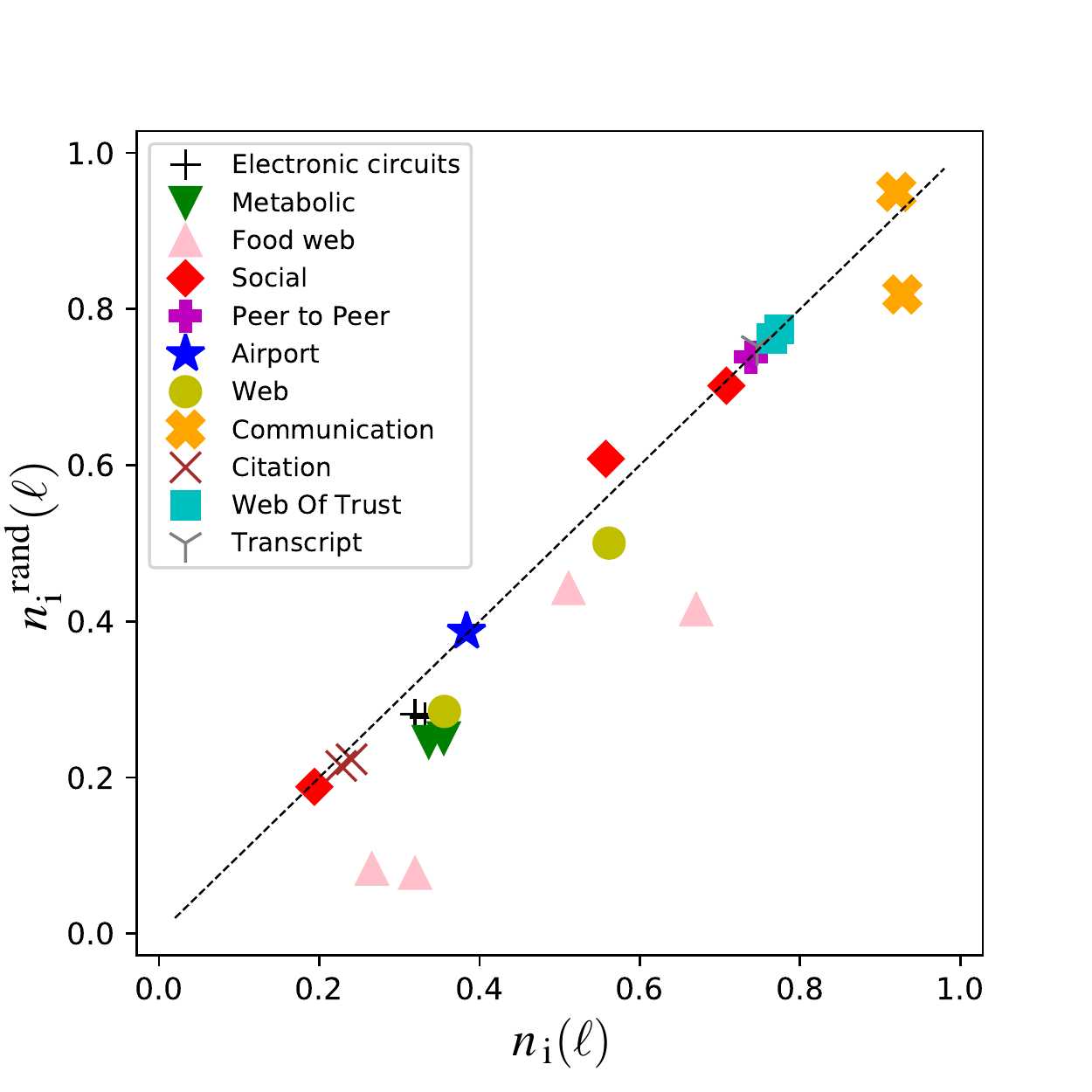}  
\caption{$\ell=2$}  
\label{fig: degree-preserve randomization LCC=2}
\end{subfigure}
\begin{subfigure}[h]{.33\textwidth}
  \centering
  \includegraphics[width=1\linewidth]{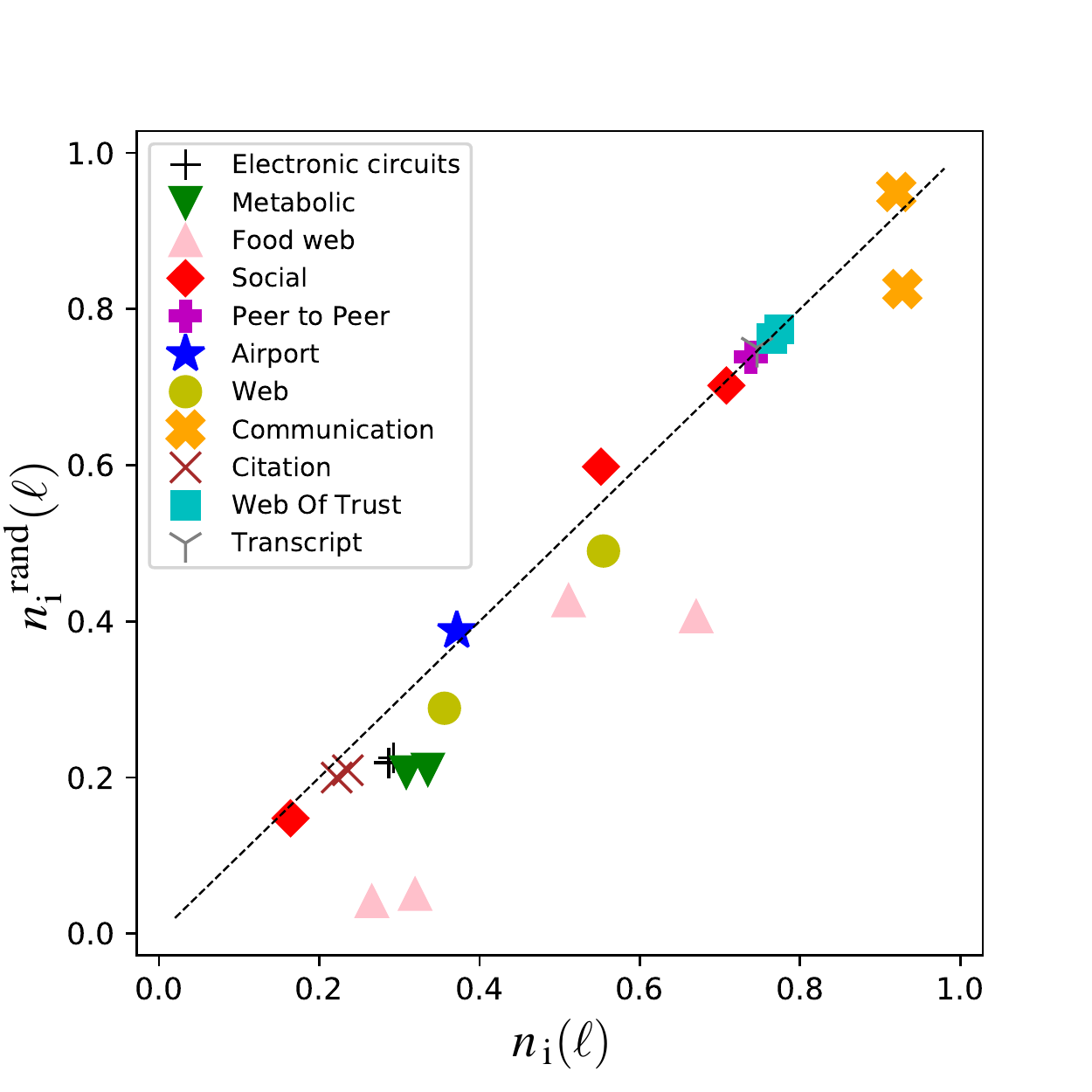}  
\caption{$\ell=3$}  
\label{fig: degree-preserve randomization LCC=3}
\end{subfigure}
\caption{\textbf{Minimum number of inputs $N_\textrm{i}(\ell)$ -- real networks.}
We measure $N_\textrm{i}(\ell)$ for a collection of real networks, randomize the networks preserving only their degree distribution and again measure the number of inputs needed for control $N^\textrm{rand}_\textrm{i}(\ell)$.
We find a strong correlation between $N_\textrm{i}(\ell)$ and $N^\textrm{rand}_\textrm{i}(\ell)$, indicating that degree distribution is an important property determining $N_\textrm{i}(\ell)$, and higher-order network properties, such as degree correlations and community structure, play a limited role.
We also observe that the correlation between $N_\textrm{i}(\ell)$ and $N^\textrm{rand}_\textrm{i}(\ell)$ becomes slightly weaker for increasing $\ell$.
Data markers indicate the average of twenty randomized networks, and for each randomization $(|\set E|/2) \times \ln {(1/\epsilon)}$ (where $\epsilon$ is between $10^{-6}$ and $10^{-7}$) number of link rewiring trials were performed~\cite{ray2012we}. Error bars are smaller than the markers.}
\label{fig: degree-preserving randomization}
\end{figure}

\subsection{Cost of the LCC constraint}

To quantify the cost of ensuring that the LCC is at most $\ell$, we calculate
\begin{equation}
C(\ell)=\frac{N_\T{i}(\ell)-N_\T{i}(\infty)}{N}, \label{eq:cost}
\end{equation}
where $N_\T{i}(\infty)$ is the number of inputs needed to control the network without any restrictions on the length of LCC.
Note that $N_\T i(\infty)$ is related but distinct from the number of driver nodes introduced by Liu et al., since in case of driver nodes the input signals are allowed to connect to multiple network nodes, while here each input is connected to exactly one node.

Figure~(\ref{fig: Cost}) shows the cost $C(\ell)$ for model networks as a function of degree heterogeneity and average degree.
A first general observation is that if $N_\T i(\infty)$ is large, i.e., the network requires a large fraction of inputs to control without restrictions on the LCC, the cost $C(\ell)$ is low.
This is natural since the LCC is the maximum distance between the set of inputs nodes and any other nodes, the more inputs there are, the shorter the expected LCC is even without explicitly tuning it.
For example, increasing degree heterogeneity increases the number inputs~\cite{liu}, correspondingly Fig.~(\ref{fig: Cost against gamma}) the cost of ensuring $l_\T{LCC} \leq 1$ for SF networks with $\gamma=4$ is $0.15$, while for $\gamma=2.1$ the cost is only $0.03$.
The role of the average degree $c$ is more complex: on one hand sparse networks require high $N_\T{i}(\infty)$, lowering the cost, on the other hand the diameter of dense networks is small making it easier to reduce the LCC.
The net effect is that $C(\ell)$ tends to zero for both $c\rightarrow 0$ and $c\rightarrow\infty$, and $C(\ell)$ has a maximum at finite $c$.
Finally, we note that increasing $\ell$ dramatically reduces $C(\ell)$.
To understand this note that complex networks are small-world, meaning that the typical number of nodes at distance $k$ from any node grows exponentially as a function of $k$.
Hence the density of the accessibility graph $\set G_\ell$ also grows exponentially with $\ell$, greatly reducing the cost.

We also calculate $C(\ell)$ for the collection of real networks (See Table~SII in the Supplementary Information), finding that for a significant number of them we can ensure the constraint on the LCC at zero cost, particularly for $\ell=3$ the majority of networks do not require additional inputs.
This means that reducing the LCC can be achieved solely by selecting the right input node set out of the possible minimum input node sets that ensure controllability and we do not need to increase their number.
Figure~(\ref{fig: cost for degree-preserve randomization}) compares the $C(\ell)$ of real networks to the cost $C^\T{rand}(\ell)$ of their degree-preserved randomized counterparts.
We find a positive correlation between $C(\ell)$ and $C^\T{rand}(\ell)$, although not as strong as the correlation observed for the total number of inputs $N_\T{i}(\ell)$, indicating that degree distribution is an important factor determining $C(\ell)$.
The strongest outliers are again the food web networks, which require close to zero cost, yet their randomized counterparts may require us to directly control more that $0.1$ fraction their nodes.  

\begin{figure}[ht]
\begin{subfigure}[h]{.49\textwidth}
  \centering
  \includegraphics[width=.8\linewidth]{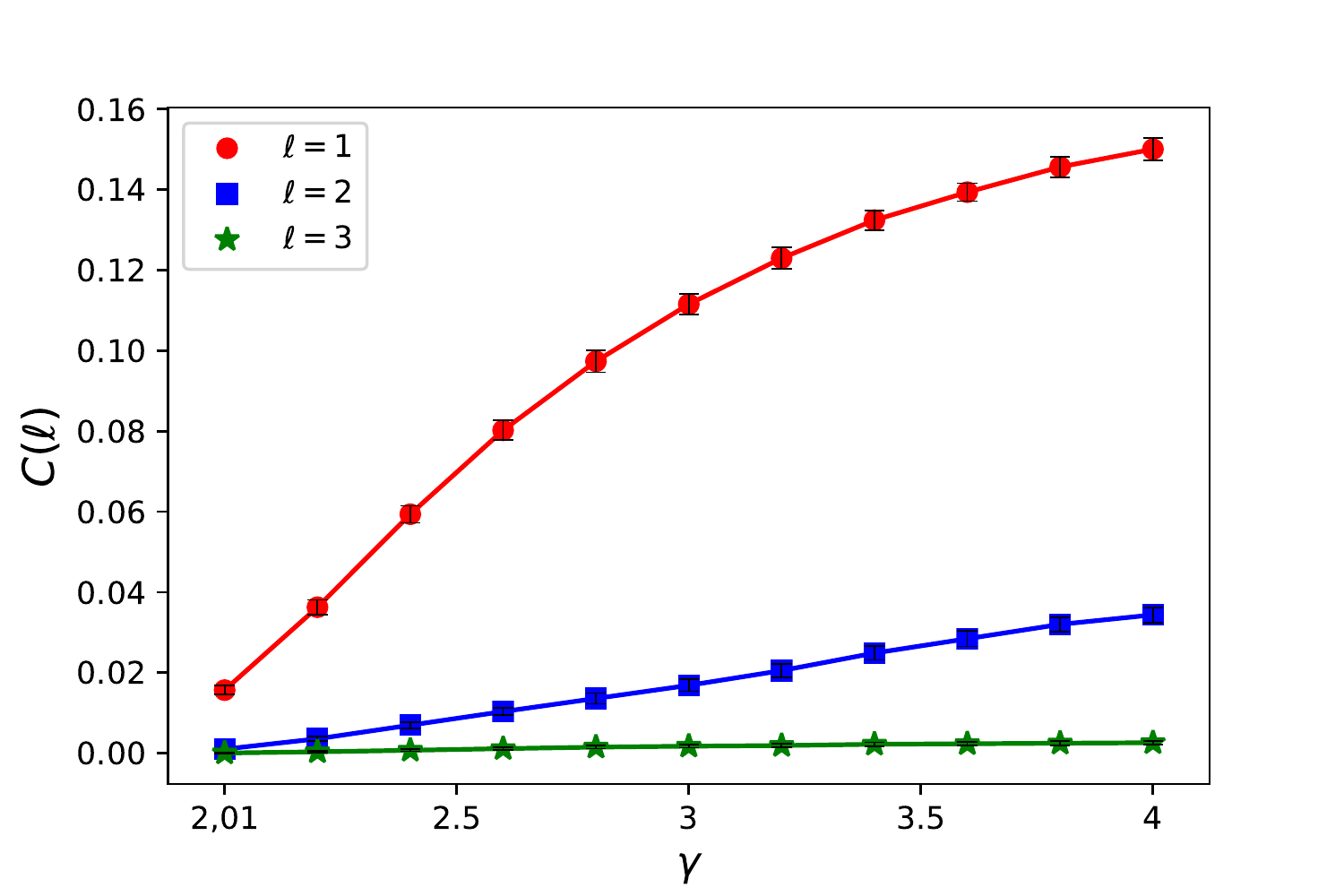}  
\caption{}  
\label{fig: Cost against gamma}
\end{subfigure}
\begin{subfigure}[h]{.49\textwidth}
  \centering
  \includegraphics[width=.8\linewidth]{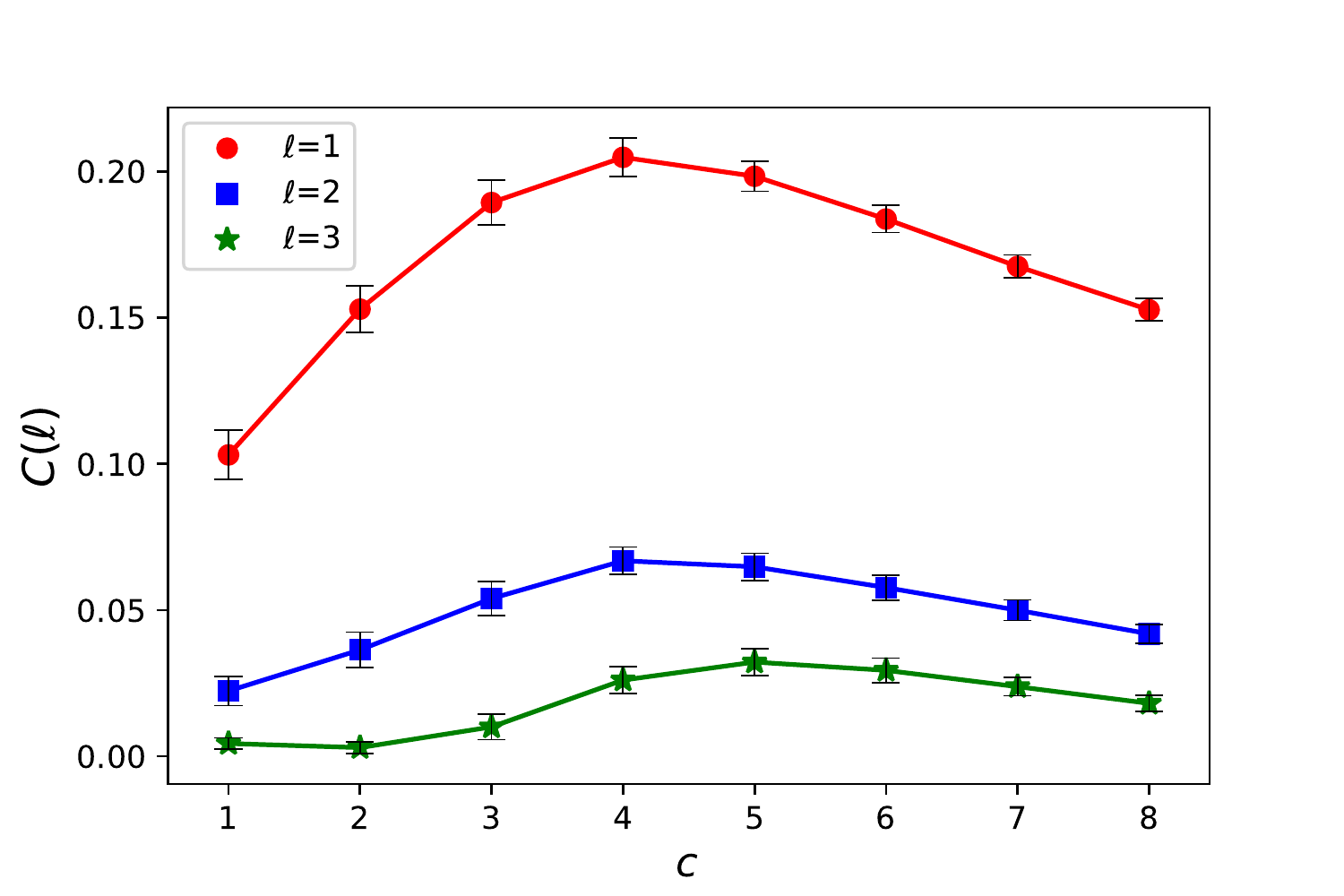}  
\caption{}  
\label{fig: Cost against c}
\end{subfigure}
\caption{\textbf{Cost of reducing the LCC -- model networks.} \textbf{(a)}~High degree heterogeneity, corresponding to low $\gamma$, reduces the diameter and increases $N_\T{i}(\infty)$, both reducing the cost. The result shows that by decreasing the heterogeneity, i.e., increasing $\gamma$ in SF networks, the cost of ensuring the LCC constraint increases.  \textbf{(b)}~For average degree $c\rightarrow 0$ the network becomes a set of isolated nodes, and as consequence $N_\T{i}(\infty)\rightarrow N$ and $C(\ell)\rightarrow 0$. On the other hand, for $c\rightarrow \infty$ the network becomes fully connected and $l_\T{LCC}\rightarrow 1$ even if there is only a single input, hence $C(\ell)\rightarrow 0$.
The figure shows SF networks with $N=10^4$ with $c=2$ in panel (a) and $\gamma=3$ in panel (b). We observe similar behavior for other parameter values. Data points are an average of one hundred independent realizations and errorbars indicate the standard error of the mean.}
\label{fig: Cost}
\end{figure}

\begin{figure}[hb]
\begin{subfigure}[h]{.33\textwidth}
  \centering
  \includegraphics[width=1\linewidth]{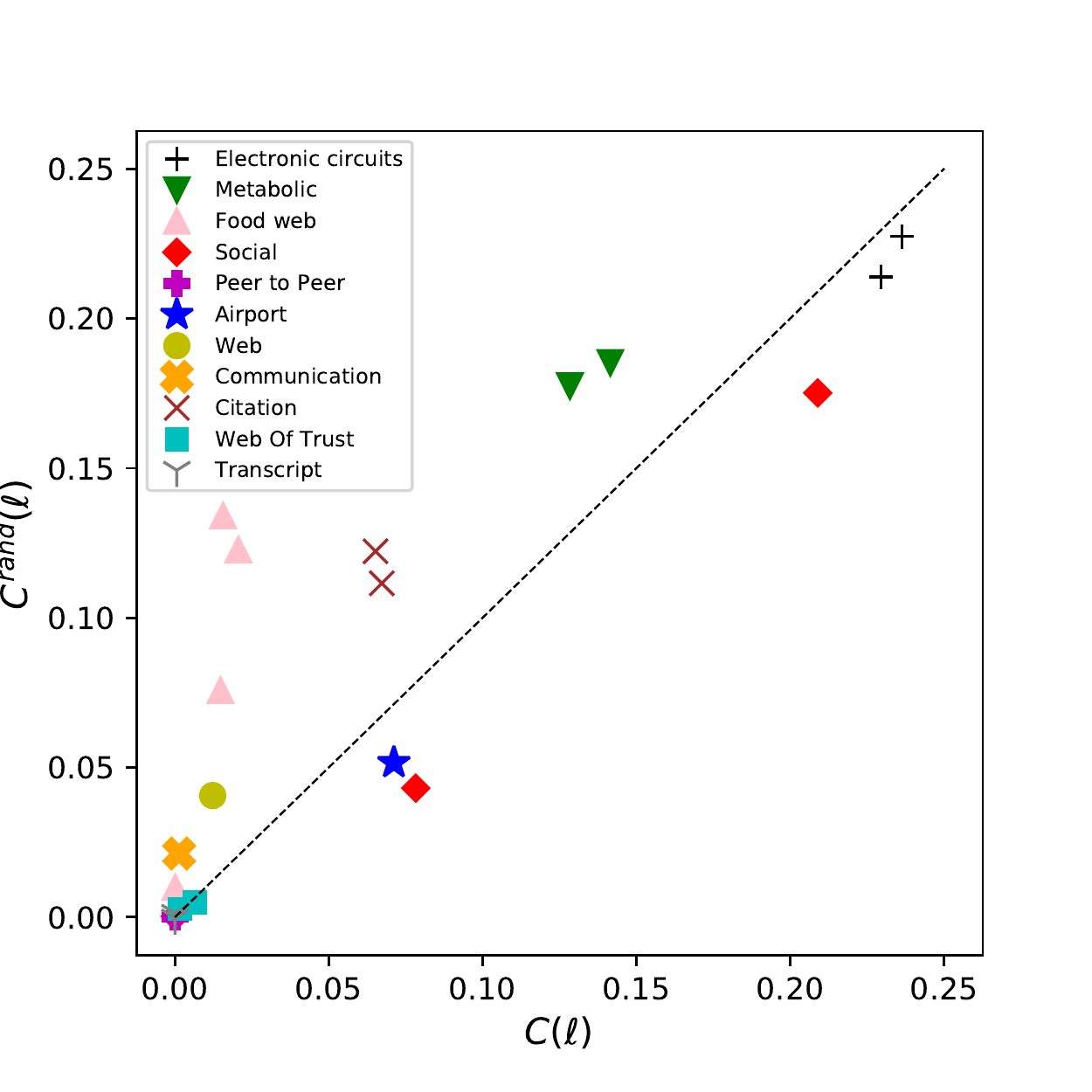}  
\caption{$\ell=1$}  
\label{fig: cost for degree-preserve randomization LCC = 1}
\end{subfigure}
\begin{subfigure}[h]{.33\textwidth}
  \centering
  \includegraphics[width=1\linewidth]{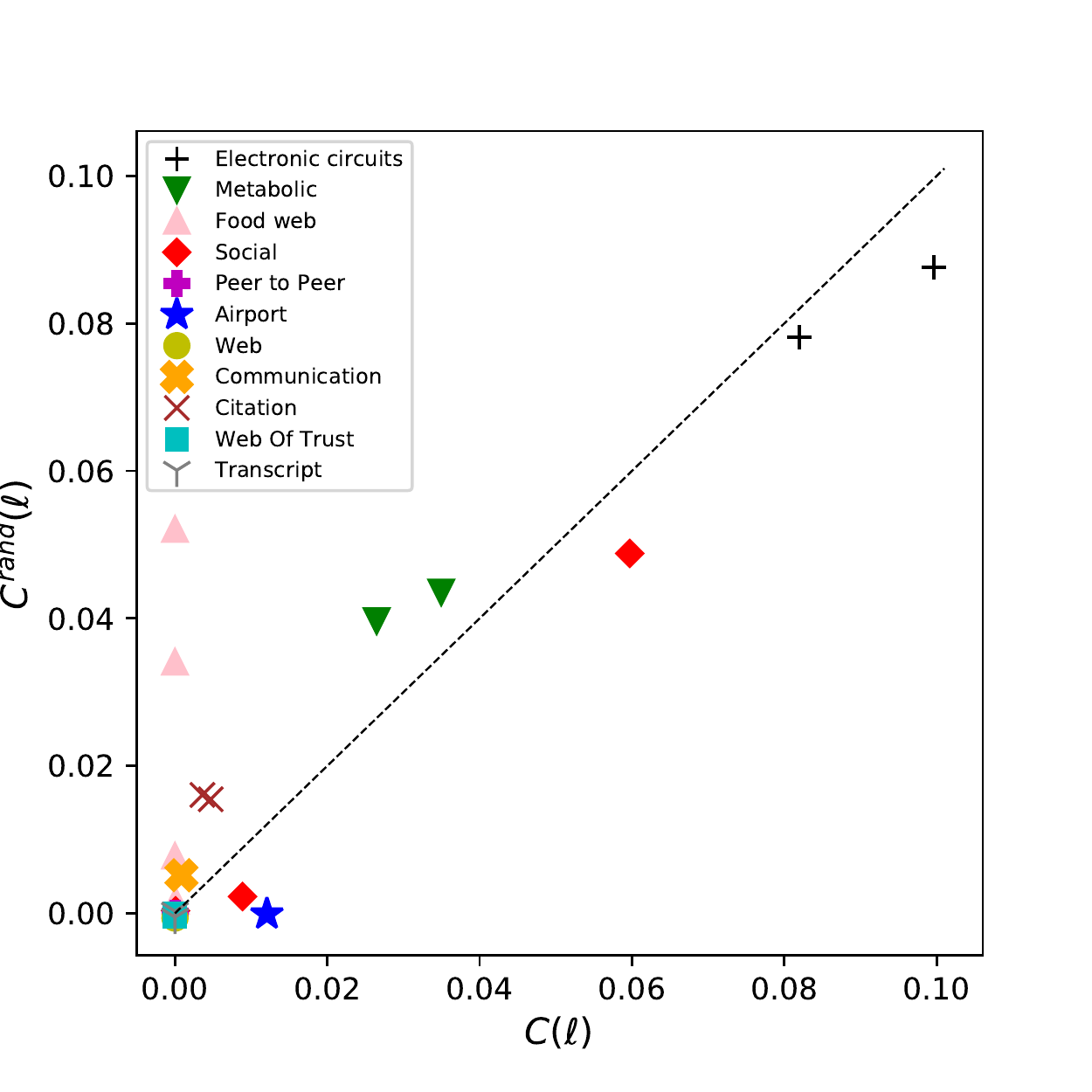}  
\caption{$\ell=2$}  
\label{fig: cost for degree-preserve randomization LCC = 2}
\end{subfigure}
\begin{subfigure}[h]{.33\textwidth}
  \centering
  \includegraphics[width=1\linewidth]{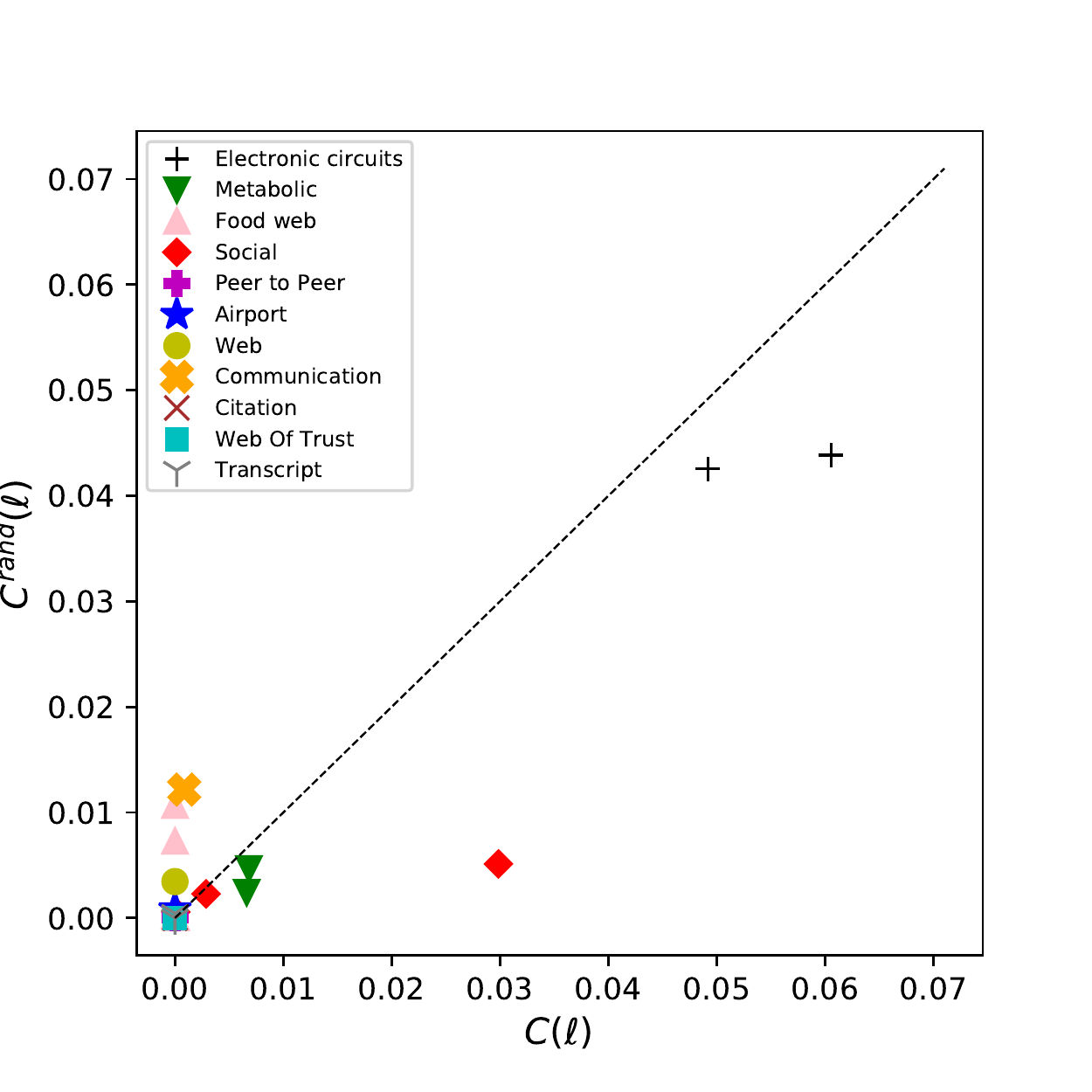}  
\caption{$\ell=3$}  
\label{fig: cost for degree-preserve randomization LCC = 3}
\end{subfigure}
\caption{\textbf{Cost of reducing the LCC -- real networks.}.We find a positive correlation between $C(\ell)$ and $C^\T{rand}(\ell)$, indicating that the degree distribution of real networks plays an important role in determining the cost of reducing the LCC.
We also find that for increasing $\ell$, an increasing number of real networks require zero cost, meaning that we can ensure ${l_\T{LCC}} \leq \ell$ by selecting the right input node set, but without increasing the number of inputs.
Data markers indicate the average of twenty randomized networks, and for each randomization $(|\set E|/2) \times \ln {(1/\epsilon)}$ (where $\epsilon$ is between $10^{-6}$ and $10^{-7}$) number of link rewiring trials were performed~\cite{ray2012we}. Error bars are smaller than the markers.}
\label{fig: cost for degree-preserve randomization}
\end{figure}

\clearpage

\section{Conclusion}\label{sec:discussion}

The minimum number of input nodes required to ensure controllability became a central metric to quantify the difficulty of controlling networks~\cite{liu,ghasemi2020diversity,ramos2014np}.
Controlling networks through a minimum set of nodes, however, often leads to unattainable control energy requirements, highlighting the need for methods of input node selection that reduce control energy.
It was recently shown that reducing the length of the longest control chain provides a useful heuristic for this which only depends on the network structure.
Motivated by this connection, we developed an algorithm to identify a minimum set of input nodes such that the network is controllable and the longest control chain is at most $\ell$.
Figure~(\ref{fig: Control energy}) shows for a directed chain network and a small real network that selecting the input nodes using this algorithm indeed reduces the control energy by orders of magnitude compared to random placement of the same number of input nodes.
Relying on the algorithm, we systematically investigated how features of complex networks affect the cost of control, showing, for example, that the cost of reducing the longest control chain decreases with increasing degree heterogeneity and is low for both very sparse and very dense networks.

Input node selection based on structural controllabilty and the longest control chain necessarily has its limitations.
For example, here we focused on fully controlling all degrees of freedom of simple directed networks with linear dynamics.
In some cases, however, complex systems are better described by generalized networks, such as hyper-graphs or multi-layer networks~\cite{posfai2016controllability,kivela2014multilayer,boccaletti2014structure,jiang2022controllability,menichetti2016control,wang2017control};
and often the goal of control is to influence only a subset of nodes~\cite{gao2014target,klickstein2017energy,guo2017constrained} or to stabilize a non-linear system in a steady state~\cite{wang2019pining,wang2002pinning,lu2013theory,fiedler2013dynamics}.
Future work may, therefore, aim to extend our method to generalized networks, or to understand how to leverage structural network features to reduce control energy beyond the assumptions of structural controllability.

\begin{figure}[h]
\begin{subfigure}[h]{.48\textwidth}
  \centering
  \includegraphics[width=.9\linewidth]{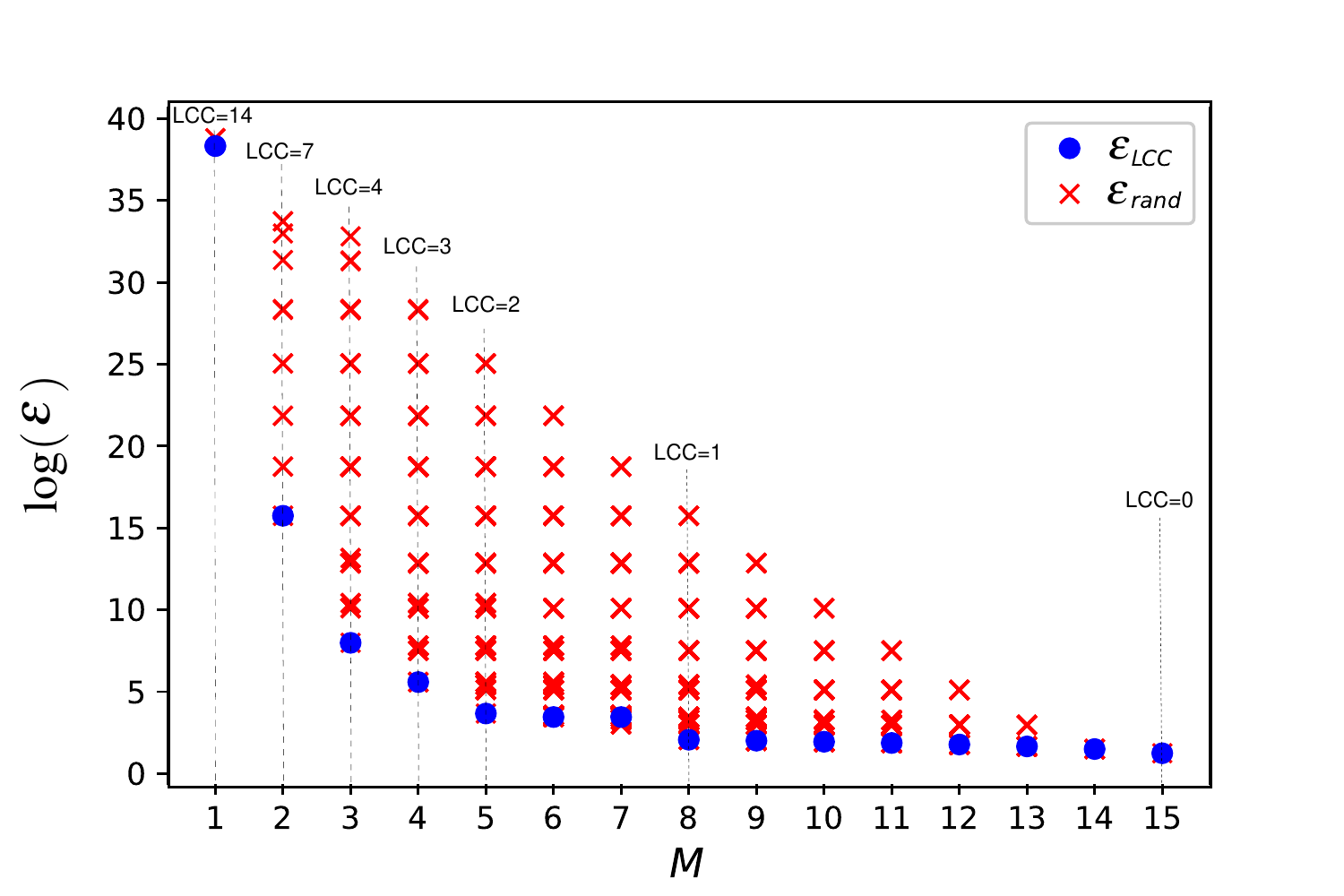}  
\caption{Directed chain}  
\label{fig: energy in a synthetic network}
\end{subfigure}
\begin{subfigure}[h]{.48\textwidth}
  \centering
  \includegraphics[width=.9\linewidth]{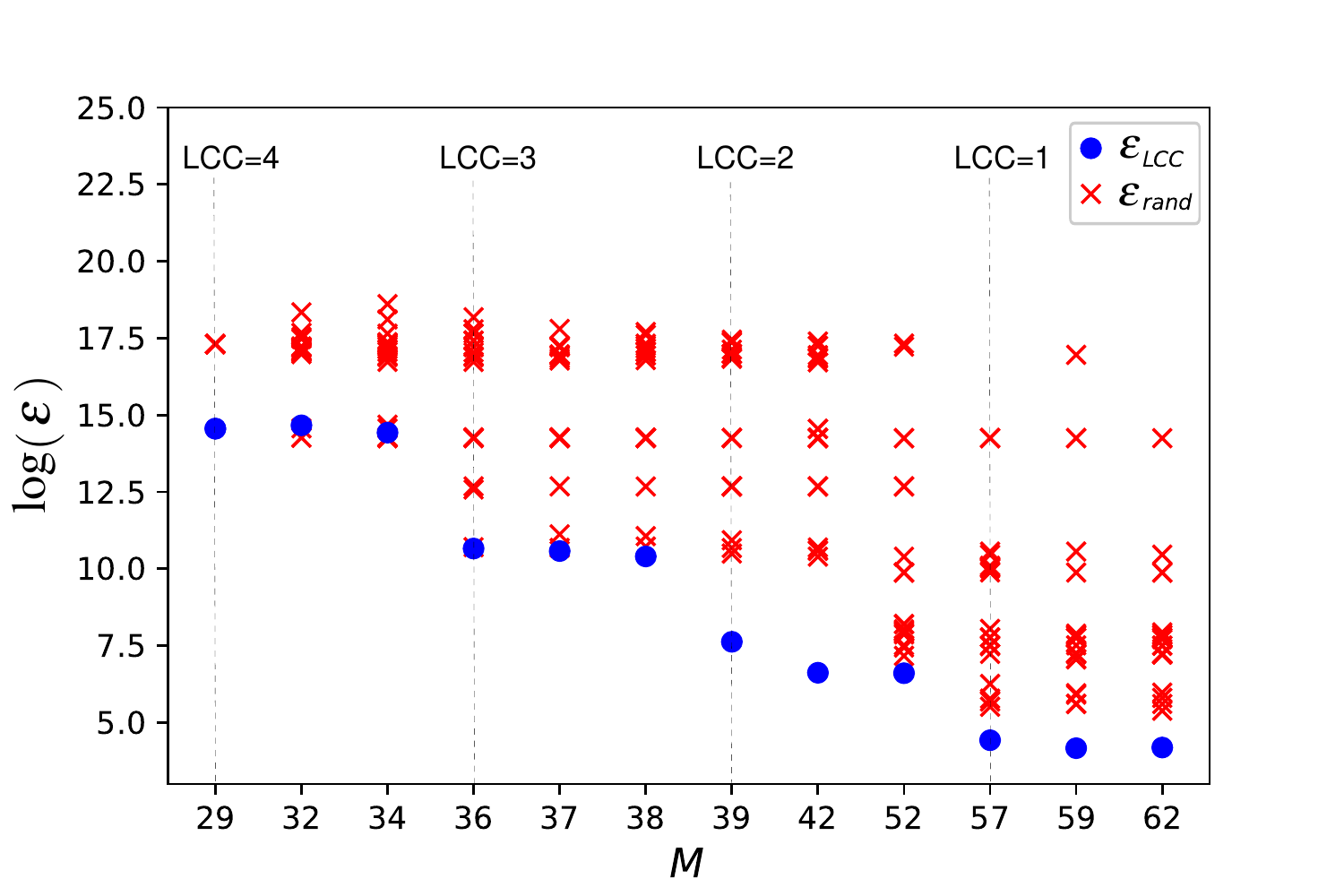}  
\caption{s208 electric circuit}  
\label{fig: energy in a real network}
\end{subfigure}
\caption{\textbf{Reduction of control energy}.
We compare the energy requirement of control using an input node selection that constrains the longest control chain (LCC) and a random input selection strategy with the same number of inputs.
Specifically, for a given number of inputs $M$, we first use the algorithm introduced in Sec.~\ref{sec: approx} to place $N_\T{i}(\ell)$ inputs, where $\ell$ is selected such that $N_\T{i}(\ell) \leq M < N_\T{i}(\ell+1)$, and we select the remaining $M - N_\T{i}(\ell)$ inputs randomly.
For random input selection, we select a minimum input node set that ensures controllability based on maximum matching and randomly select the remaining inputs.
The energy requirement of control is quantified by the mean energy $\langle \varepsilon \rangle= \textrm{tr}(\M W_\textrm{B}^{-1})$ with control time $t_\T{f}=1$.
We calculated the control energy for \textbf{(a)}~a directed chain with $15$ nodes and \textbf{(b)}~a real network $s208$ representing an electric circuit. Vertical dashed lines indicate $M$ values where $\ell$ changes.
The energy requirement of LCC constrained input selection (blue markers) can reduce the control energy by orders of magnitude compared to random selection (red markers).}
\label{fig: Control energy}
\end{figure}

\subsection*{Author contributions}

A.G., S.A. conceived the project. S.A., M.P., A.G. designed the research. S.A. performed simulations and analyzed the empirical data. S.A., M.P., A.G. wrote the manuscript.

\subsection*{Acknowledgments}

M.P. acknowledges funding from ERC grant No.~810115-DYNASET. A.G. gratefully acknowledges funding
from the Alexander von Humboldt Foundation (Ref. 3.4 - IRN -1214645 - GF-E) for visiting the University of Passau, Germany. 

\clearpage

\bibliography{ref}

\end{document}

\title{Supplementary Information -- Controllability of complex networks: input node placement restricting the longest control chain}

\author{Samie Alizadeh}
\affiliation{Department. of Computer Engineering, K. N. Toosi University of Technology, Tehran, Iran}
\author{M\'arton P\'osfai}
\affiliation{Department of Network and Data Science, Central European University, Vienna, Austria} 
\author{Abdorasoul Ghasemi}
\affiliation{Department. of Computer Engineering, K. N. Toosi University of Technology, Tehran, Iran}

\maketitle

\tableofcontents

\newpage

\section{Computational complexity of the LCC-constrained minimum input problem}\label{app:np-complete}

In the following we prove that the computational complexity of the LCC-constrained minimum input problem belongs to the NP-complete class by showing that its corresponding decision problem is both NP and NP-hard\cite{cormen2009introduction}.
To prove that it is NP, we provide a polynomial-time algorithm to verify for every node set $\set S$ whether $\set S$ is a valid input node set or not.
Proving that it is NP-hard involves reducing the minimum dominating set problem, a known NP-complete problem, to the LCC-constrained minimum input problem, where the reduction is carried out in polynomial time.

\begin{theorem}
For a given directed network $\set G(\set V, \set E)$ and positive integer $\ell>0$, a node set $\set S\in \set V$ is a valid input node set if
\begin{enumerate}[(i)]
\item there exists a matching in the bipartite representation $\set B$ of network $\set G$, such that $\set S$ is the set of unmatched nodes;
\item and $\set S$ is a dominating set in the accessibility graph $\set G_\ell$ of network $\set G$.
\end{enumerate}
Determining whether a valid input node set $\set S \subseteq \set V$ exists such that $\lvert\set S\rvert = M>0$ is an NP-complete problem.
\end{theorem}

\begin{proof}

To prove that the problem is NP, we present the following algorithm to check the validity of an input node set $\set S$.
To check if $\set S$ satisfies the matching condition (i), we obtain a pruned bipartite network $\set B^\prime$ by taking $\set B$ and removing all vertices from $\set V^-$ that correspond to vertices in $\set S$.
If there exists a prefect matching in $\set B^\prime$ then there exists a matching in $\set B$ such that only nodes in $\set S$ are unmatched.
We can check whether a perfect matching exists in a bipartite network using, for example, the Hopkroft-Karp algorithm which has worst case runtime $O(\sqrt{\lvert\set V\rvert} |\lvert \set E\rvert)$~\cite{hopcroft1973n}.

To check the accessibility condition (ii), we fist construct the accessibility graph $\set G_\ell$ by calculating the distance between all node pairs in $\set G$, which can be done using breath-first search in at most $O\left(\lvert\set V\rvert (\lvert\set V\rvert + \lvert\set E\rvert)\right)$ steps~\cite{cormen2009introduction}.
Verifying that $\set S$ is a dominating set in $\set G_\ell$ can be done, for example, by iterating over all edges in $\set G_\ell$ which takes at most $O(\lvert\set V\rvert^2)$ steps.
Consequently, verifying a solution can be done in a polynomial time.

To prove NP-hardness, we reduce the NP-hard minimum dominating set (MDS) problem to our problem~\cite{hartmanis1982computers}.
For a given directed network $\set G(\set V, \set E)$ a set of nodes $\set D\in\set V$ is a dominating set if each node in $\set V$ is either in $\set D$ or has a link pointing at it that starts from a node in $\set D$.
To reduce the MDS to our problem, we first create an augmented network $\set G^\prime$ by adding a self-loop to all nodes in $\set G$, solving the LCC-constrained minimum input problem with $\ell=1$ in $\set G^\prime$ is equivalent to solving the MDS in $\set G$.
To see this, note that the set of self-loops provide a perfect matching in $\set G^\prime$; therefore we only have to check that the minimum input set satisfies the accessibility condition.
The accessibility graph $\set G_1^\prime$ is the same as $\set G^\prime$ and self-loops do not affect the dominating sets; therefore the minimum input node set in $\set G^\prime$ for $\ell=1$ is also a minimum dominating set in $\set G$.

Therefore, the computational complexity of the LCC-constraint minimum input problem belongs to the NP-complete class.
\end{proof}

\clearpage

\section{Efficient integer linear programming formulation of the LCC-constrained minimum input problem}\label{app:alternative ILP}

Analyzing the performance of the integer linear programming (ILP) formulation (8) in the main text shows that a na\:ive implementation of the matching results in poor performance and run-time quickly increases with the number of links in the network (Fig.~\ref{fig: Explored nodes}).
The performance can be improved by a constant factor using the graph-cycling formulation introduced by Ref.~\cite{iudice2015structural}.

The basic idea of this approach comes from this idea that a matching partitions the digraph into disjoint paths, see figure (\ref{fig: Graph cycling model}). Thus, they partition the digraph into disjoint cycles to find a matching. 

We start from a directed graph $\set G(\set V, \set E)$ (Fig.~\ref{fig: Graph cycling model}a). We begin by creating an augmented graph $\set G^\prime(\set V^\prime, \set E^\prime)$ by adding an auxiliary node $x$ to $\set G$, representing the external control signals.
We connect this node to all network nodes by a pair of out-going and in-coming links (Fig.~\ref{fig: Graph cycling model}b).
Then, we partition $\set G^\prime$ with cycles in a way that cycles are only allowed to overlap at node $x$ (Fig.~\ref{fig: Graph cycling model}c).
Links that participate in cycles in $\set G^\prime$ form a matching in $\set G$, and the disjoint cycles correspond to disjoint paths in the matching.
Therefore the links exiting the auxiliary node as a part of cycle point at input nodes, which are unmatched nodes in the matching (Fig.~\ref{fig: Graph cycling model}d).
Note that there are different ways to partition $\set G^\prime$ into disjoint cycles.
To find a maximum matching, we aim to partition the $\set G^\prime$ with a minimum number of cycles.

To define the associated ILP formulation, we define binary variables $y_{i \to j} \in \{0,1\}$ corresponding to each link $(i,j)$ in the augmented graph. Here, $y_{i \to j}=1$ means that $(i,j)$ is a part of cycle partitioning, otherwise, $y_{i \to j}=0$. To ensure that the solution is a cycle partition, each node $v$ in the network is forced to have exactly one in-coming and one out-going link:
\begin{subequations}
\begin{align}
\sum_{j \in {\set V_v}^+} y_{v \to j} = 1 \label{eq: inbounding constraint}\\
\sum_{i \in {\set V_v}^-} y_{i \to v} = 1 \label{eq: outgoing constraint}
\end{align}
\end{subequations}
Then, among all possible cycle partitions, we select the one that creates the minimum number of cycles. The number of cycles can be determined by the number of outgoing links from the auxiliary node that participate in cycles.
Thus, we minimize the $\sum_{j \in \set V} y_{x \to j}$ objective function. Note that the number of out-going and in-coming links of the auxiliary node must be equal, thus they have to satisfy the $\sum_{j \in \set V} y_{x \to j} = \sum_{i \in \set V} y_{i \to x}$  constraint.
Next, we should ensure the accessibility. Thus, for each node $i$, the auxiliary node must connect to at least one node of set $\set V_i^{\ell}$, where $\set V_i^{\ell}$ is the set of nodes from where we can reach node $i$ in at most $\ell$ steps.
This is enforced by the $\sum_{j \in \set V_i^{\ell}} y_{x \to j} \geq 1$ constraint.
Putting it all together, we obtain the following ILP formulation:

\begin{subequations}
\label{eq: problem1-ILP-graph cycling}
\begin{align}
& \hspace{3mm} \min_{y_{i \to j} \in \set E'} \hspace{5mm} \displaystyle\sum_{j \in \set V} y_{x \to j} \label{eq: problem1-ILP-graph cycling-a} \\
\textrm{subject to} & \nonumber \\
& \hspace{3mm} y_{i \to j} \in \{0,1\} \hspace{17mm} \forall \,\, {(i \to j) \in \set E'} \label{eq: problem1-ILP-graph cycling-b}\\
& \hspace{2mm} \displaystyle\sum_{j \in {\set V_v}^+} y_{v \to j} = 1 \hspace{15.7mm} \forall \, {v = 1,...,N} \label{eq: problem1-ILP-graph cycling-c}\\
& \hspace{2mm} \displaystyle\sum_{i \in {\set V_v}^-} y_{i \to v} = 1 \hspace{16.8mm} \forall \, {v = 1,...,N} \label{eq: problem1-ILP-graph cycling-d}\\
& \hspace{2.2mm} \displaystyle\sum_{j \in \set V_v^{\ell}} y_{x \to j} \geq 1 \hspace{17.2mm} \forall \, {v = 1,...,N} \label{eq: problem1-ILP-graph cycling-e} \\
& \hspace{2.5mm} \displaystyle\sum_{j \in \set V} y_{x \to j} = \displaystyle\sum_{i \in \set V} y_{i \to x} \label{eq: problem1-ILP-graph cycling-f}
\end{align}
\end{subequations}

Figure~\ref{fig: Explored nodes} compares the performance of the na\:ive ILP formulation (8) and the graph-cycling implementation in Eq.~(\ref{eq: problem1-ILP-graph cycling}).
We measure the run-time by the number of nodes searched in the branch-and-bound tree to solve the problem.
We find that although the run-time for both formulations grows exponentially, the graph-cycling implementation is significantly faster, allowing us to explore larger network instances.

\begin{figure}[h]
\begin{subfigure}
\centering
\includegraphics[width=.87\linewidth]{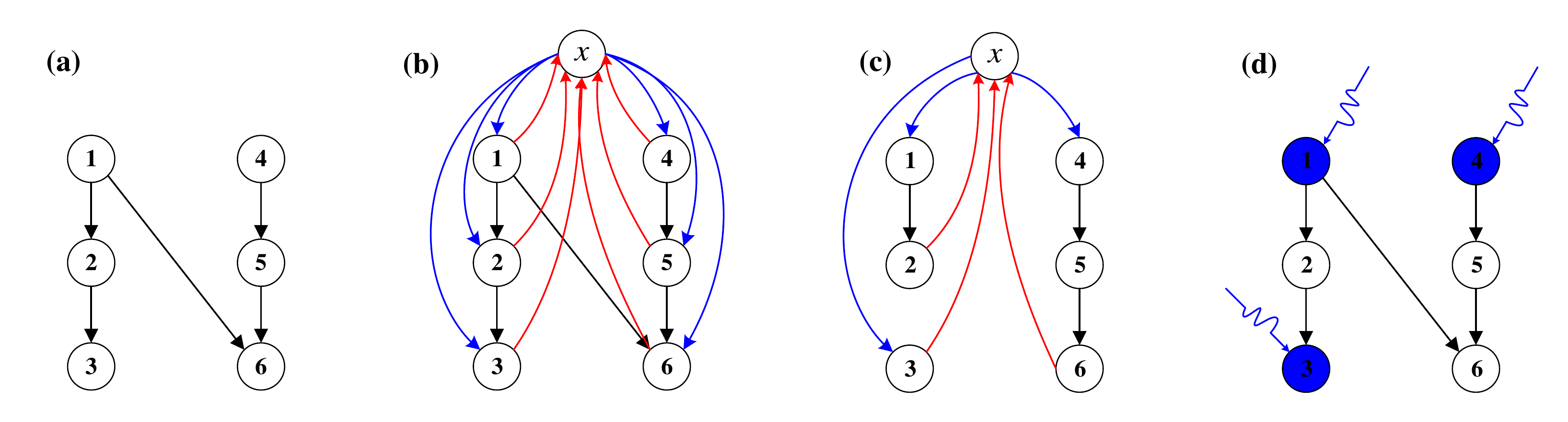}
\end{subfigure}
\caption{\textbf{Graph-cycling ILP formulation.} \textbf{(a)}~A small directed network. \textbf{(b)}~We create an augmented graph by adding an auxiliary node and connecting it to all nodes via a pair of in-coming and out-going links. \textbf{(c)}~The partitioned graph for $\ell=1$. Here, the directed graph is partitioned into three disjoint cycles. \textbf{(d)}~Input nodes are the nodes that have a link pointing at them that originates from the auxiliary node and is part of a cycle.}
\label{fig: Graph cycling model}
\end{figure}

\begin{figure}[!h]
\centering
\includegraphics[scale=.45]{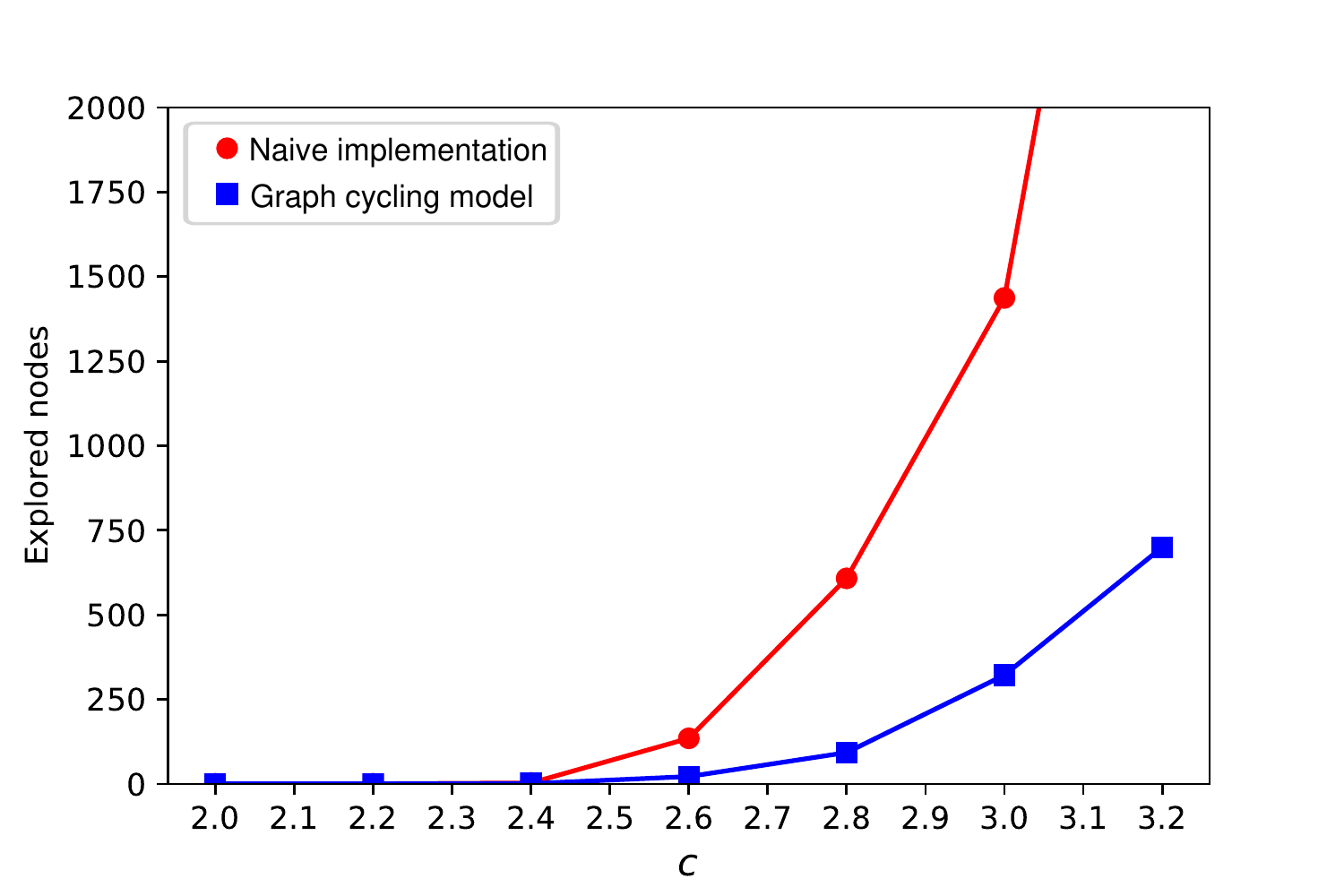}
\caption{\textbf{Comparing the performance of the ILP implementations}. We measure the run-time of the models by the number of nodes searched by the ILP solver in the branch-and-bound tree. We explored ER networks with $N = 500$ nodes. Each data point represents an average of 100 independent instances.}
\label{fig: Explored nodes}
\end{figure}

\clearpage

\section{Real networks}

\begin{table}[h]
\centering
\small
\caption{The collection of real networks. We provide the number of nodes $N$, number of links $L$, average degree $c$, degree heterogeneity $H = \textrm{max}(H^{in}, H^{out})$, where $H^{in/out}$ can be defined by $ H^{in/out}=\frac{1}{cN^2}\sum_i \sum_j (k_i^{in/out}-k_j^{in/out})$\cite{liu}, $c$ is a constant, and $k_i^{in/out}$ shows input/output degree of node.}
\begin{tabular}{ll|rrrrc}
\multicolumn{2}{c|}{Networks} & $\quad\quad N$ & $\quad\quad\quad\quad L$ & $\quad\quad c$ & $\quad\quad H$ & ~~~Ref.~~~ \\ \hline\hline

\multirow{2}{*}{Electric circuits}   &  s208 & 122 & 189  & 1.5 & .63  & \cite{milo2002network}    \\
                                                       &  s838 & 512 & 819  & 1.5 & .64  & \cite{milo2002network}   \\ \hline

\multirow{4}{*}{Food Web}        & Baywet      & 128   & 2,106  & 16.4  & .8      & \cite{baird1998assessment}  \\
					& Mangwet    & 97    & 1,492  & 15.3  & .98    &  \cite{martinez1991artifacts} \\
					& Ythan         & 135  & 601      & 4.4    & 1.42  & \cite{ulanowicz2005network} \\
                                                     & Littlerock    & 183  & 2,494   & 13.6  & 1.36  & \cite{ulanowicz2005network}  \\ \hline

\multirow{2}{*}{Transcript}           & E.Coli    & 418  & 519     & 1.2  &  1.8      & \cite{milo2002network}  \\
                                                       & Yeast     & 688  & 1,079  & 1.5   &  1.84   & \cite{milo2002network}  \\ \hline

\multirow{2}{*}{Web of Trust}           & CentralCoast   & 943  & 1,227     & 1.3  &  1.7      & \cite{levy2018innovation}  \\
                                                             & Napa-rev        & 646  & 926       & 1.4  &  1.7  & \cite{levy2018innovation}  \\ \hline

\multirow{2}{*}{Metabolic}           & C.Elegans    & 1,173  & 2,864  & 2.4  & .7      & \cite{jeong2000large}  \\\
                                                       & Yeast           & 1,511  & 3,833  & 2.5  & .78    & \cite{jeong2000large} \\\hline

\multirow{1}{*}{Airport}           & USairport    & 1,574  & 28,236  & 17.9  & 1.51      & \cite{opsahl2010node} \\ \hline

\multirow{2}{*}{Web}           & polblogs-rev    & 1,224  & 19,022  & 15.5  & 1.51      & \cite{jeong2000large}  \\
			               & Stanford    & 281,903  & 2,312,497  & 8.2  & 1.26      & \cite{leskovec2009community}  \\\hline

\multirow{3}{*}{Social}           & prisoninmate    & 67  & 182  & 2.7  & .81      & \cite{jeong2000large}  \\
			                 & Epinions1         & 75,879  & 508,524  & 6.7  & 1.71      & \cite{richardson2003trust}  \\
			                 & WikiVote          & 7,115  & 139,311  & 14.5  & 1.69      & \cite{leskovec2010signed} \\\hline

\multirow{2}{*}{Communication}           & Email   & 265,214  & 420,045     & 1.5  &  1.58      & \cite{leskovec2007graph}  \\
                                                             & WikiTalk        & 2,394,385  & 5,021,410       & 2.1  &  .14  & \cite{leskovec2010signed}  \\ \hline

\multirow{2}{*}{Citation}                    & HepPh   & 34,546  & 420,045     & 12.2  &  1.58      &\cite{leskovec2005graphs}  \\
                                                             & HepTh   & 27,770  & 352,807            & 12.7  &  1.7  & \cite{leskovec2005graphs}  \\ \hline

\multirow{1}{*}{Peer to Peer}           & Gnutella31    & 62,586  & 147,892  & 2.3  & .88      & \cite{ripeanu2002mapping} \\ \hline
                                                  
\end{tabular}
\label{table:realnet_properties}
\end{table}

\clearpage
                                              
\begin{table}[h]
\small
\centering
\caption{The result of created core and required cost in real networks for different LCC constraints.}
\begin{tabular}{c | lll | lll | lll}
     &   \multicolumn{3}{l}{$\ell=1$} & \multicolumn{3}{l}{$\ell=2$} & \multicolumn{3}{l}{$\ell=3$}\\ \hline
\multicolumn{1}{c|}{Networks}  & ~$n_\T{i}(\ell)$~ & ~$n_\T{core}$~ & ~$C(\ell)$~  & ~$n_\T{i}(\ell)$~ & ~$n_\T{core}$~ & ~$C(\ell)$~ & ~$n_\T{i}(\ell)$~ & ~$n_\T{core}$~ & ~$C(\ell)$~ \\ \hline \hline
\multicolumn{1}{l|}{s208}   & .46  & .45 & .22  & .31 & .38 & .08 &  .29 & .30 & .04    \\
\multicolumn{1}{l|}{s838}   & .45  & .48 & .22  & .33 & .46 & .09 &  .29 & .42 & .06    \\
\multicolumn{1}{l|}{Baywet}   & .28  & .91 & .03  & .28 & .67 & .007 &  .26 & .38 & .007    \\
\multicolumn{1}{l|}{Mangwet}   & .30  & .84 & .07  & .31 & .55 & .01 &  .31 & .31 & 0    \\
\multicolumn{1}{l|}{Ythan}   & .52  & .24 & .01  & .51 & .04 & 0 &  .51 & .04 & 0    \\
\multicolumn{1}{l|}{Littlerock}   & .66  & .85 & .01  & .65 & .85 & .01 &  .65 & .85 & .005   \\
\multicolumn{1}{l|}{E.Coli}   & .74  & 0 & 0  & .74 & 0 & 0 &  .74 & 0 & 0    \\
\multicolumn{1}{l|}{Yeast}   & .82  & .001 & 0  & .82 & .001 & 0 &  .71 & .001 & 0    \\
\multicolumn{1}{l|}{CentralCoast}   & .77  & .003 & .006  & .76 & 0 & 0 &  .76 & 0 & 0    \\
\multicolumn{1}{l|}{Napa-rev}   & .77  & 0 & .001  & .77 & 0 & 0 &  .77 & 0 & 0    \\
\multicolumn{1}{l|}{C.Elegans}   & .28  & .65 & .14  & .33 & .64 & .03 &  .30 & .50 & .03    \\
\multicolumn{1}{l|}{USairport} & .44 &   .76 &  .07   & .38 & .13 & .01 & .37 & .09 & .001     \\ 
\multicolumn{1}{l|}{polblogs}   & .37  & .37 & .01  & .35 & .006 & 0 &  .35 & .005 & 0    \\
\multicolumn{1}{l|}{Stanford}    & .45 &   .81  &   .14    & .37 &  .72 & .05 & .35 & .61 &  .05   \\
\multicolumn{1}{l|}{prisoninmate}   & .34  & .76 & .19  & .19 & .74 & 0.04 &  .16 & .71 & .02    \\
\multicolumn{1}{l|}{Epinions1}   & .63  & .12 & .08  & .56 & .04 & .01 &  .55 & .02 & 0    \\
\multicolumn{1}{l|}{WikiVote}    &  .67 &  .0007  & .007   & .70 & 0 & 0.0001 & .70 & 0 &  0      \\ 
\multicolumn{1}{l|}{Email}      & .92  & .01     &  .001    &  .92 & 0.00001 & .0008 & .92 & 0 &  0   \\
\multicolumn{1}{l|}{WikiTalk} & .96  &   .001  &  .0001   & .96  & 0 & 0 & .96 & 0 &  0  \\ 
\multicolumn{1}{l|}{HepPh} & .30  & .23 & .07  & .24 & .12 & .008 & .23 &  .09 &  .003   \\
\multicolumn{1}{l|}{HepTh} & .29 &   .40 &  .07   & .23 & .23 & .01 & .22 & .20 & .008     \\ 
\multicolumn{1}{l|}{Gnutella31} & .73  &  .60     &   .0001    & .73  & .0005 & 0  &.73 & .0005 &  0 \\ \hline
\end{tabular}
\label{table:realnet_results}
\end{table}

\clearpage

\bibliography{ref}